\RequirePackage{fix-cm}
\documentclass[smallextended]{svjour3}       
\smartqed  
\usepackage{graphicx}
\usepackage{amsfonts}
\usepackage{amsmath}
\usepackage{amssymb}
\usepackage{algorithm}
\usepackage{algpseudocode}
\usepackage{color}
\definecolor{darkblue}{rgb}{0,0,.6}
\definecolor{darkgreen}{rgb}{0,.4,0}
\usepackage[colorlinks,citecolor=darkgreen,linkcolor=darkblue]{hyperref}
\usepackage{anysize}
\usepackage{url}
\usepackage{multirow}
\usepackage{authblk}

\begin{document}

\title{The Generalized Independent and Dominating Set Problems on Unit Disk Graphs\thanks{Preliminary
version of this paper appeared in FAW, 2018}}



\author{Sangram K. Jena \and Ramesh K. Jallu \thanks{Ramesh K. Jallu was supported by the Czech Science Foundation, grant number GJ19-06792Y, and by institutional support RVO:67985807} \and Gautam K. Das \and Subhas C. Nandy}

\authorrunning{Sangram et al.} 

\institute{ Sangram K. Jena \at
              Department of Mathematics\\Indian Institute of Technology Guwahati \\
              \email{sangram@iitg.ac.in}           
           \and
           Ramesh K. Jallu \at
              The Czech Academy of Sciences, Institute of Computer Science \\ 
              \email{jallu@cs.cas.cz}
              \and
             Gautam K. Das  \at
              Department of Mathematics\\Indian Institute of Technology Guwahati  \\ 
              \email{gkd@iitg.ac.in}
              \and
               Subhas C. Nandy \at
              Indian Statistical Institute, Kolkata \\ 
              \email{nandysc@isical.ac.in}
}

\date{Received: date / Accepted: date}

\maketitle

\begin{abstract}
In this article, we study a generalized version of the maximum independent set and minimum dominating set problems, namely, the maximum $d$-distance independent set problem and the minimum $d$-distance dominating set problem on unit disk graphs for a positive integer $d>0$. We first
show that the maximum $d$-distance independent set problem and the minimum $d$-distance dominating set problem belongs to NP-hard class. Next, we propose a simple polynomial-time constant-factor approximation algorithms and  PTAS for both the problems.
\end{abstract}

\section{Introduction}
\label{sec:intro}
The \emph{independent set problem} is one of the well known classical combinatorial optimization problems in graph theory due to its many important applications, including but not limited to networks, map labeling, computer vision, coding theory, scheduling, clustering. Like \emph{independent set problem}, \emph{dominating set problem} is also an important well studied combinatorial optimization problem in graph theory. The dominating set problem has a wide range of applications including wireless networking, facility location problems etc.

Given an unweighted graph $G = (V,E)$, a non-empty subset of pairwise non-adjacent 
vertices of $G$ is known as an \emph{independent set} of $G$. The maximum independent set problem
asks to find an independent set of maximum size in a given 
unweighted graph $G$, and such a set is called as \emph{maximum independent 
set} (MIS) of $G$. For an integer $d \geq 2$, a \emph{distance-$d$ 
independent set} (D$d$IS) of an unweighted graph $G=(V,E)$ is an independent set $I$ 
of $G$ such that the shortest path distance (i.e., the number of edges on a shortest path) 
between every pair of vertices in $I$ is at least $d$. For a 
given unweighted graph $G$, the objective of the maximum distance-$d$ independent set 
problem is to find a D$d$IS of maximum 
cardinality in $G$. A D$d$IS of maximum possible size 
is called as \emph{maximum distance-$d$ independent set} (MD$d$IS). 
Observe that the D$d$IS problem is a generalization of the MIS problem and 
in fact for $d=2$, the D$d$IS problem and MIS problem are the same.\\
In a simple unweighted graph $G=(V,E)$, dominating set is defined as a set of vertices $V'\subseteq V$ such that for each vertex $u\in V$, either (i) $u\in V'$, or (ii) there exist $v\in V'$ such that $v$ is a neighbor of $u$. The dominating set of minimum cardinality in a graph $G$ is called the minimum dominating set (MDS) of $G$. The objective of MDS problem is to find a dominating set of minimum cardinality in $G$. We define a generalized version of MDS problem as distance-$d$ dominating set (D$d$DS) problem. A D$d$DS for an integer $d\geq 1$ in a simple unweighted graph $G=(V,E)$  is defined as a set of vertices $V'\subseteq V$ such that, for each vertex $u\in V$, either (i) $u\in V'$, or (ii) $ v \in V'$ such that, the shortest path distance between $u$ and $v$ is at most $d$. The objective of the \emph{minimum distance-$d$ dominating set} (MD$d$DS) problem is to find a D$d$DS of minimum cardinality in a given graph $G$.

Given a set $P = \{p_1,p_2,\ldots,p_n\}$ of $n$ points in the plane, a \emph{unit disk graph} (UDG) 
corresponding to the point set $P$ is a simple graph $G=(V,E)$ such that 
$V = P$, and 
$E = \{(p_i,p_j) \mid d(p_i,p_j) \leq 1\}$, where $d(p_i,p_j)$ denotes the Euclidean distance 
between $p_i$ and $p_j$. In other words, a unit disk graph is an intersection graph of 
disks of unit diameter centered at the points in $P$.

An algorithm for a minimization (resp. maximization) problem 
is said to be a  {\it $\rho$-factor approximation algorithm} 
if for every instance of the problem the algorithm produces 
a feasible solution whose value is within a factor of $\rho$ (resp. at least a factor of $\frac{1}{\rho}$) of the 
optimal solution value and runs in polynomial-time of the input size. 
Here, $\rho$ is called \emph{the approximation 
factor} or \emph{approximation ratio} of the algorithm and the optimization problem is said to have a $\rho$-factor 
approximation algorithm. A \emph{polynomial-time approximation scheme} 
(PTAS) for an optimization problem is a collection of algorithms $\{{\cal A_\epsilon}\}$ such 
that for a given $\epsilon >0$, ${\cal A_\epsilon}$ is a $(1 + \epsilon)$-factor 
approximation algorithm in case of minimization problem ($(1 - \epsilon)$ in 
case of maximization). The running time of ${\cal A_\epsilon}$ is required 
to be polynomial in the size of the problem depending on ${\epsilon}$.
\section{Related work}\label{sec:rel_work}
The MIS problem is known to be NP-hard for general graphs \cite{garey2002} including many sub-class of planar  
graphs, namely planar graphs of maximum degree 3 \cite{garey1977}, 
planar graphs of large girth \cite{murphy1992}, cubic planar graphs \cite{garey1974}, triangle-free graphs 
\cite{poljak1974},  $K_{1,4}$-free graphs \cite{minty1980}, etc. 

Tarjan and Trojanowski \cite{tarjan1977} presented a naive algorithm for finding maximum independent set in a graph having $n$-vertices in $O(2^\frac{n}{3})$ time. Later Robson \cite{robson1986} improved the complexity  to $O(2^{0.276n})$. Xiao and Nagamochi \cite{xiao2017} gave a better bound for finding the maximum independent set problem in $1 . 1996^n n^{O(1)}$ time and in polynomial space. For the graphs with maximum degree $6$ and $7$, they gave algorithms to find MIS, which run in  $1.1893^nn^{O(1)}$ time and $1.1970^nn^ {O(1)}$ time,  respectively. Johnson et al. \cite{johnson1988} presented an algorithm, which produces lexicographic ordering of all maximal independent sets of a graph having polynomial delay between two successive independent sets with exponential space complexity. In that paper, they also proved that there is no such polynomial-delay algorithm exists for generating all maximal independent sets in reverse lexicographic order, unless P=NP.

Andrade \cite{andrade2012} gave the first local search algorithm for finding independent set of a graph. Later the problem studied on pseudo-disks in the plane by Chan and Har-Peled \cite{chan2012}. They analysed the problem for both weighted and unweighted cases and  gave a PTAS via local-search algorithm for unweighted case and for weighted case, they gave a constant-factor approximation by an LP based  rounding scheme.

In general, the MIS problem cannot be approximated within a constant factor unless P=NP \cite{arora1998}. 
However, the problem is polynomially 
solvable for bipartite graphs, outerplanar graphs, 
perfect graphs, claw-free graphs, chordal graphs, etc. \cite{gavril1972,hsu1981}.
The MIS problem is well studied on UDGs too and is shown to be NP-hard  \cite{clark1990}. 
Unlike in general graphs, the problem admits approximation algorithms 
\cite{marathe1995,halldorsson1995,matsui2000,Das2015,jallu2016,nandy2017} and 
approximation schemes \cite{erlebach2005,nieberg2005,Das2015,jallu2016}.

The distance-$d$ independent set (D$d$IS) problem, for any fixed $d \geq 3$, is known to be NP-hard for bipartite graphs \cite{chang1984} and 
planar bipartite graphs of maximum degree 3 \cite{eto2014}. It is also known that getting an 
$n^{\frac{1}{2}-\epsilon}$-factor approximation 
result, for any $\epsilon > 0$, on bipartite graphs is NP-hard (this result 
also holds for chordal graphs when $d \geq 3$ is an odd number)
\cite{eto2014}. 
The problem is polynomially solvable for some intersection 
graphs, such as interval graphs, trapezoid graphs, and circular arc graphs \cite{agnarsson2003}.
If the input graph is restricted to be a chordal graph, then the problem is solvable in polynomial time for any even $d \geq 2$; 
on the other hand, the problem is NP-hard for any odd $d \geq 3$ \cite{eto2014}.
Eto et al. \cite{eto2016} studied the problem 
on $r$-regular graphs and planar graphs. The authors showed that for $d \geq 3$ and 
$r \geq 3$, the D$d$IS problem on $r$-regular graphs is APX-hard, and proposed 
$O(r^{d-1})$ and $O(\frac{r^{d-2}}{d})$-factor approximation algorithms. When $d=r=3$, they 
enhanced their $O(\frac{r^{d-2}}{d})$-factor result to a 2-factor approximation 
result (later, the approximation factor is improved to 1.875 \cite{eto2017}).
Finally, they proposed a PTAS in case of planar graphs. 
Montealegre and Todinca studied the problem in graphs with few minimal separators \cite{montealegre}.\\

The minimum dominating set (MDS) problem is known to be NP-hard \cite{garey2002}. Raz and Safra \cite{raz1997sub} proved the inapproximability for the MDS problem by showing that there does not exist any approximation algorithm better than $O(n \log n)$-factor approximation algorithm unless P=NP.
Due to lack of scope in better approximation result in general graphs, researchers tried geometric version of MDS problem to get better approximation factor.

The MDS problem is studied on UDG and proved to be NP-hard \cite{clark1990}. Nieberg and Hurink \cite{nieberg2005ptas1} proved that the problem admits a $(1+\epsilon)$-factor approximation algorithm for $0 < \epsilon \leq 1$. By assigning $\epsilon = 1$, a $2$-approximation algorithm can be obtained, which is fastest. The running time of this algorithm is $O(n^{81})$ \cite{de2013approximation}. Gibson and Pirwani
\cite{gibson2010algorithms} gave a PTAS for MDS problem
of arbitrary size disk graph, which runs in  $n^{O(\frac{1}{\epsilon^2})}$ time.

For the MDS problem in unit disk graphs, a 5-factor approximation algorithm is proposed by Marathe et al. \cite{marathe1995simple} in $O(n^2)$ time. Carmi et al. \cite{carmi2008polynomial} proposed a 5-factor
approximation algorithm for the MDS problem in arbitrary size disk graph.
Fonseca et al. \cite{FonsecaFSM12} improved the factor to $\frac{44}{9}$ for MDS problem in unit disk graph by using  the local improvement technique, which runs in  $O(n \log n)$ time.  De et al. \cite{de2013approximation} proposed a 12-factor approximation algorithm for the MDS problem in unit disk graph with running time $O(n \log n)$. In the same paper, they improvised the approximation factors to  4-factor, and 3-factor in time  $O(n^8 \log n)$, and $O(n^{15}\log n)$ respectively. Carmi et al. \cite{rameshanna} improved the time complexity of 4-factor approximation algorithm to $O(n^6\log n)$. They also proposed a simple 5-factor approximation algorithm in $O(n \log k)$ time  for this problem, where $k$ is the size of the output. In the same paper, they also proposed $\frac{14}{3}$-factor, 3-factor and $\frac{45}{13}$-factor approximation algorithm for MDS problem in unit disk graphs with time complexity $O(n^ 5 \log n)$, $O(n^{11} \log n)$ and $O(n^{10} \log n)$ respectively. Finally, with the help of shifting lemma, they proposed a $\frac{5}{2}$-factor approximation algorithm in $O(n^{20} \log n)$ time.
\subsection{Our work}\label{sec:our_work}
In this paper, we study the distance-$d$ independent set (D$d$IS) problem and distance-$d$ dominating set (D$d$DS) problem on unit disk graphs, where position of the disk centers are known. We call the geometric version of D$d$IS problem as \emph{geometric distance-$d$ independent set} (GD$d$IS) problem and D$d$DS problem as \emph{the geometric distance-$d$ dominating set} (GD$d$DS) problem. 
We show that the decision version of the GD$d$IS problem (for $d \geq 3$) is NP-complete on unit disk graphs (refer to Section \ref{sec:hardness}). We proposed a simple 4-factor approximation algorithm for GD$d$IS problem in Section \ref{sec:aa}, and a PTAS for this problem in Section \ref{sec:ptas}. We also show that the decision version of the GD$d$DS problem (for $d \geq 2$) is NP-complete on unit disk graphs (refer to Section \ref{sec:hardness-1}). We proposed a simple 4-factor approximation algorithm for GD$d$DS problem in Section \ref{sec:apprx}, and a PTAS for GD$d$DS problem in Section \ref{sec:apprx_schm}. Finally, we conclude the paper in Section \ref{sec:conclusion}.
\section{The GD$d$IS Problem on Unit Disk Graphs} \label{sec:hardness}
For an integer $d\geq 3$, we define the GD$d$IS problem as follows:
\begin{itemize}
\item []\emph{Given an unweighted unit disk graph $G=(V,E)$ corresponding 
to a point set $P = \{p_1,p_2,\ldots,p_n\}$ in the plane, find a maximum 
cardinality subset $I\subseteq V$, such that for every pair of vertices 
$p_i,p_j \in I$ the length (number of edges) of the shortest path between $p_i$ and $p_j$ in $G$ is  at least $d$.}
\end{itemize}

For a fixed constant $d\geq 3$, the decision version D(GD$d$IS) of the GD$d$IS problem is defined as follows:
\begin{description}
 \item[Input.] An unweighted unit disk graph $G=(V,E)$ defined on a point set $P$ 
 and a positive integer $k \leq |V|$.
 \item[Question.] Does there exist a distance-$d$ independent set of size at least $k$ in $G$?
\end{description}
\begin{lemma}\label{lem:np}
The problem belongs to class NP due to following reason.
\end{lemma}
\begin{proof}
 Given any subset $V' \subseteq V$, we can verify whether each pair of vertices in $V'$ is $d$-distance independent or not in polynomial time using Floyd-Warshall's all-pair shortest path algorithm \cite{cormen}.
\end{proof}
Now, we show that the D(GD$d$IS) ($d \geq 3$) problem belongs to NP-hard class by polynomial time reduction of D(GD$d$IS) from \emph{distance-$d$ independent set problem 
($d \geq 3$) on planar bipartite graphs with girth\footnote{the length of a smallest cycle in the graph} at least $d$ and maximum degree 3}, which is known to be NP-hard.

\begin{description}
 \item[\textsc{Decision Version of D$d$IS Problem on Planar Bipartite Graphs}]
 \item[Input.] An unweighted planar bipartite graph $G=(V,E)$ with girth at least $d$ and maximum vertex degree 3, 
 and a positive integer $k \leq |V|$.
 \item[Question.] Does there exist a distance-$d$ independent set of size at least $k$ in $G$?
\end{description}
In \cite{eto2014}, it has been shown that the distance-$d$ independent set problem on planar bipartite graphs with maximum degree 3 is NP-hard using polynomial time reduction of it from the distance-2 independent set problem on planar cubic graphs, which is known to be NP-hard \cite{johnson1976some}.
In fact, the reduced graph in their reduction has girth at least $d$ and 
hence the distance-$d$ independent set problem on planar bipartite graphs 
with maximum degree 3 and girth at least $d$ is NP-hard.

Our reduction is based on the concept of planar 
embedding of planar graphs. The following lemma is very useful in our reduction.

\begin{lemma}\cite{clark1990}\label{embedding}
 A planar graph $G=(V,E)$ with maximum degree 4 can be embedded in the plane using $O(|V|^2)$ 
 area in such a way that its vertices are at integer coordinates and its edges are 
 drawn using axis-parallel line segments at integer coordinates
 (i.e., edges lie on the lines $x=i_1, i_2, \ldots$ and/or $y=j_1, j_2, \ldots$, 
 where $i_1, i_2, \ldots, j_1, j_2, \ldots$ are integers). 
\end{lemma}

\begin{corollary}\label{cor:graph_embed}
 Let $G=(V,E)$ be a planar bipartite graph with maximum degree 3 and  
 girth at least $d ~(d \geq 3)$. $G$ can be 
 embedded on a grid in the plane, whose each grid cell is of size  $d \times d$, 
 so that its vertices lie at points of the form $(i*d,j*d)$ and its edges are drawn using a 
 sequence of consecutive line segments drawn on the vertical lines of 
 the form $x=i*d$ and/or horizontal lines of the form  
 $y=j*d$, for some integers $i$ and $j$ \emph{(see Figure. \ref{fig:embedding})}.
\end{corollary}
\begin{figure}[!ht]
  \centering
  \begin{minipage}{.45\textwidth}
  \centering
  \includegraphics[scale=0.85]{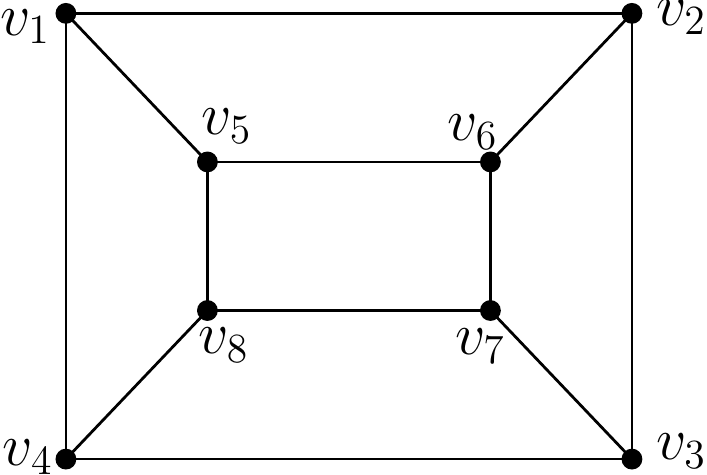} \\ {(a)}
  \end{minipage}
  \begin{minipage}{.45\textwidth}
  \centering
  \includegraphics[scale=0.85]{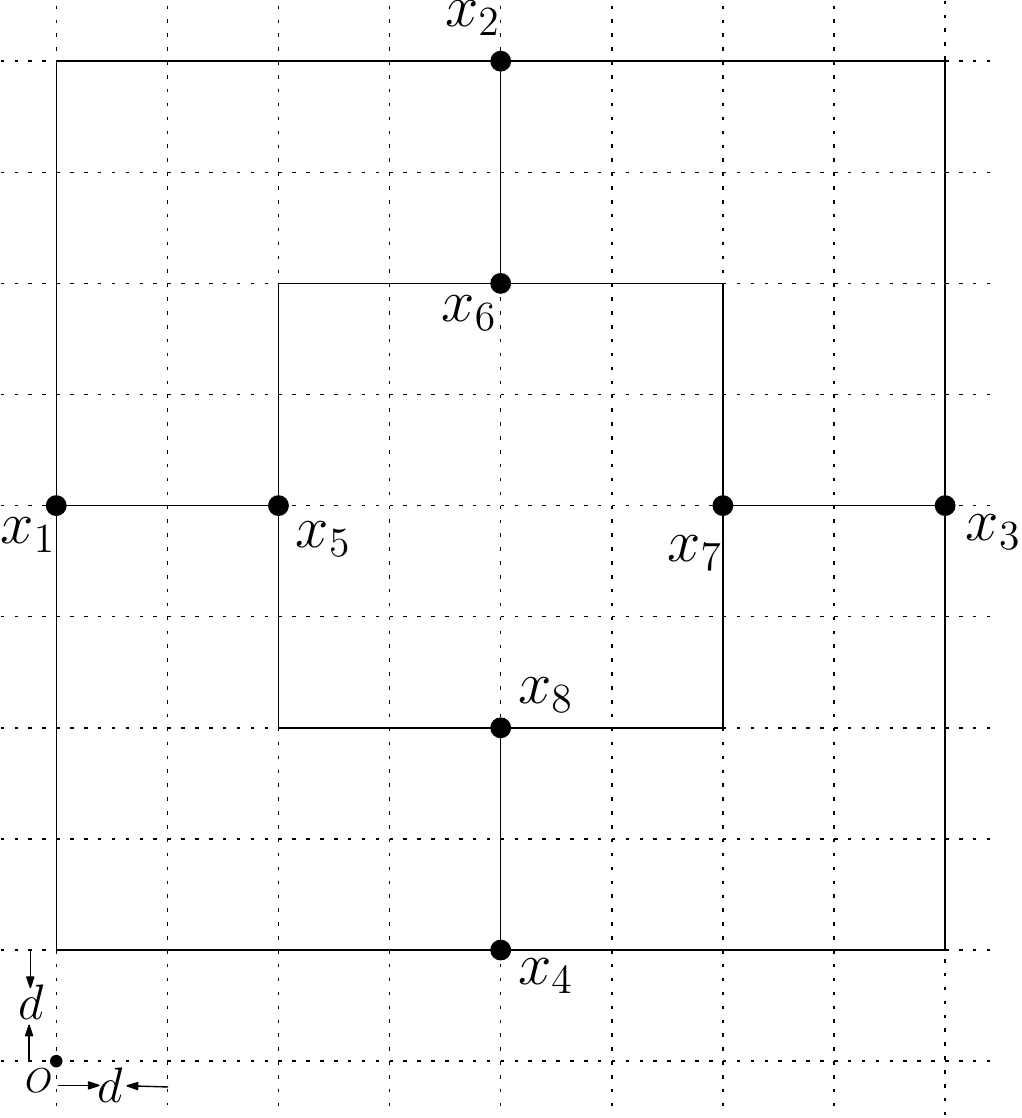}\\
  {(b)}
  \end{minipage}
  \begin{minipage}{.45\textwidth}
  \centering
  \hspace{-2.0cm}
  \includegraphics[scale=0.85]{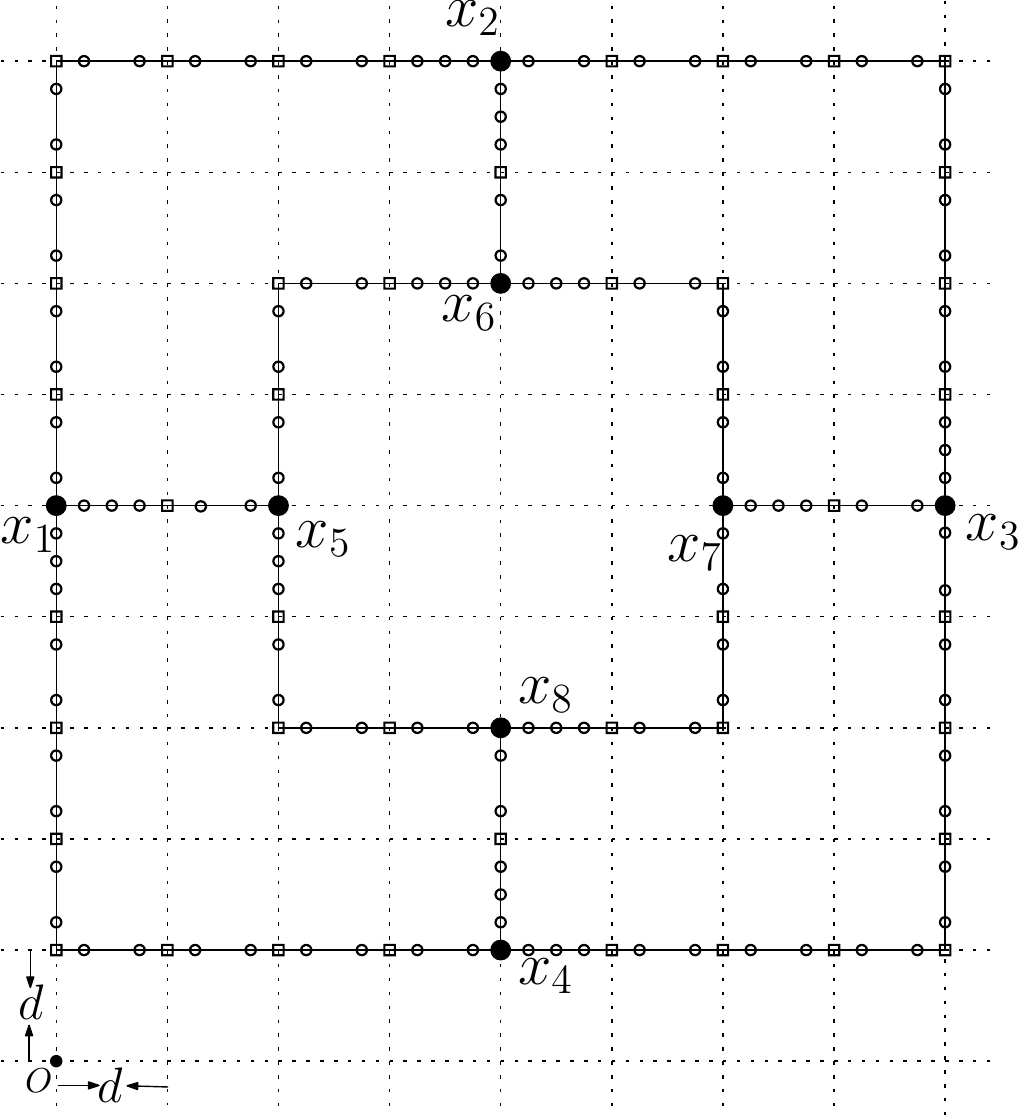}\\
  {(c)}
  \end{minipage}
  \begin{minipage}{.45\textwidth}
  \centering
  \hspace{2.5cm}
  \includegraphics[scale=0.85]{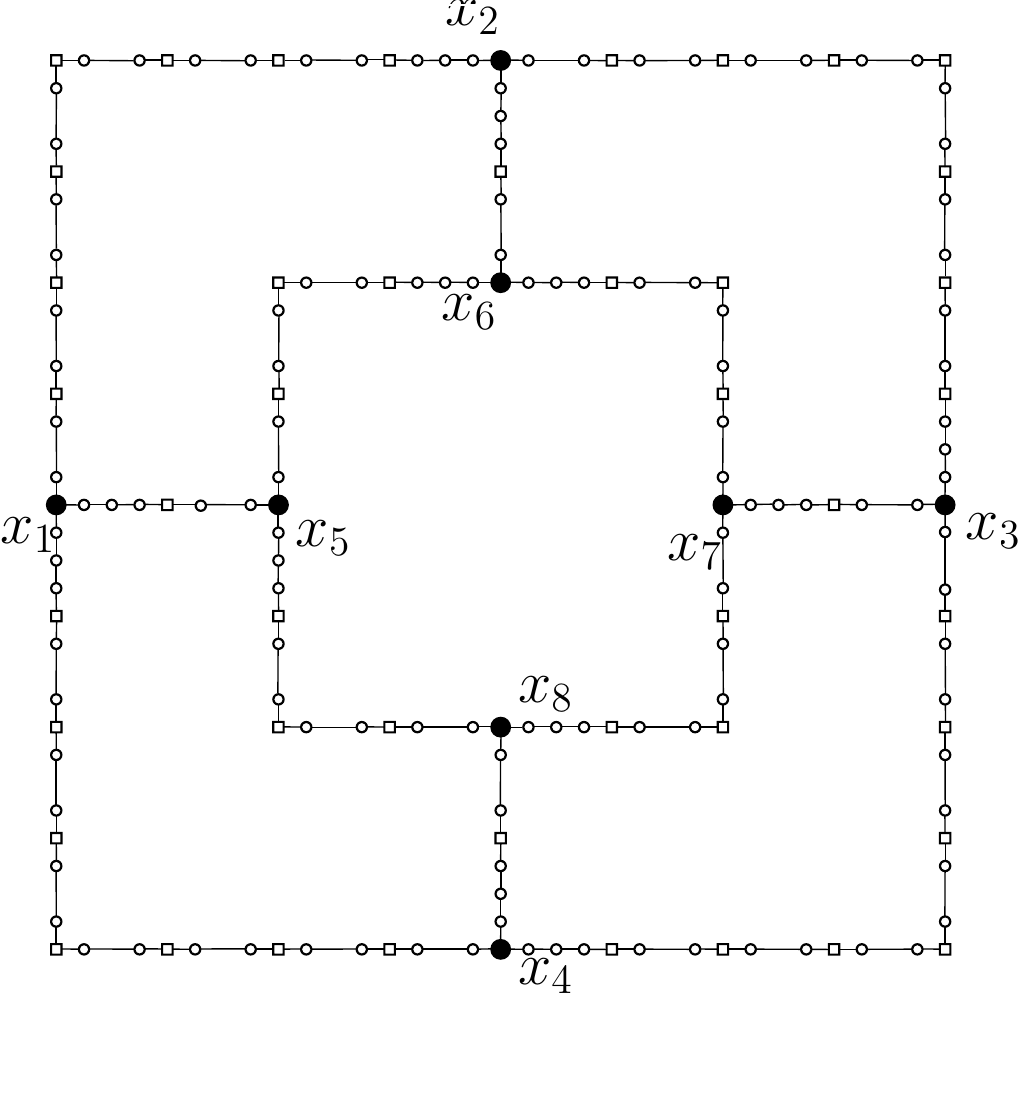}\\
  {(d)}
  \end{minipage}
  \caption{(a) A planar bipartite graph $G$ of maximum degree 3, (b) its embedding $G'$ on a grid of cell size
  $3\times 3$, (c) adding of extra points to $G'$, (d) the obtained UDG $G''$.}\label{fig:embedding}
  \end{figure}
  \begin{proof}
   Lemma \ref{embedding} suggests that, any planar graph $G$ of maximum 
degree 4 can be embedded on a grid in the plane so that; 
\begin{enumerate}
 \item Each vertex $v_i$ of $G$ is associated with a point with integer coordinates in the plane.
 \item An edge of $G$ is represented as a sequence of alternating 
 horizontal and/or vertical line segments drawn 
 on the grid lines. For example, see edges $(v_1,v_4)$ or $(v_2,v_6)$ in Figure. \ref{fig:embedding}(a). 
 The edge $(v_1,v_4)$ is drawn as a sequence of four vertical line segments and four horizontal 
 line segments in the embedding (see $(x_1,x_4)$ in Figure. \ref{fig:embedding}(b)). Similarly, 
 the edge $(v_2,v_6)$ is drawn as a sequence of two vertical line segments in the embedding.
 \item No two sets of consecutive line segments correspond to two distinct 
 edges of $G$ have a common point unless the edges incident at a vertex in $G$. 
\end{enumerate}
\end{proof}  

This kind of embedding is known as \emph{orthogonal drawing} of a graph. Biedl and 
Kant \cite{biedl1998} gave a linear time algorithm that produces an orthogonal 
drawing of a given graph with the property that the number of bends along each edge is at most 2.
 
Let $G=(V,E)$ be an arbitrary instance of D$d$IS for planar bipartite graph having maximum 
degree three and girth at least $d$. Let $V=\{v_1,v_2, \ldots ,v_n\}$ and $E=\{e_1,e_2,\ldots,e_m\}$. 
We denote the shortest path distance between two vertices $v_i$ and $v_j$ in $G$ by 
$d_G(v_i,v_j)$ and $v_i, v_j$ are said to be distance-$d$ independent in $G$ if and only if $d_G(v_i,v_j) \geq d$.

We construct a graph $G'=(V',E')$ by embedding $G$ on a grid in which each cell is of size $d \times d$ 
as described in Corollary \ref{cor:graph_embed}. 
Let $V'=\{x_1,x_2,\ldots,x_n\}$ be the vertices in $G'$ corresponding to $v_1,v_2,\ldots,v_n$ 
in $G$. The coordinate of each member in $V'$ is of the form $(d.i, d.j)$, where $i,j$ are 
integers, and shown using big dots in Figure. \ref{embedding}(c). Let $\ell$ be the number of line 
segments used for drawing all the edges in $G'$. To make $G'$ a UDG we introduce a set $Y$ of 
extra points on the segments used to draw the edges of $G'$. Thus, the set of points in $V'$ 
(hereafter denoted by $X$) together with $Y$ form a UDG $G''$.
Let $(x_i,x_j)$ be an edge in $G'$ 
corresponding to the edge $(v_i,v_j)$ in $G$ and has $\ell'$ grid segments. We introduce $\ell'd$ 
points on the polyline denoting the edge $(x_i,x_j)$ in such a way that (i) after adding the 
extra points, the length of the path from $x_i$ to $x_j$ is exactly $\ell'd+1$, (ii) a point is placed at
each of the co-ordinates of the form $(d \cdot i, d \cdot j)$, where $i$ and $j$ are integers (shown using small squares in Figure. 
\ref{embedding}(c)), (iii) the segment adjacent to the point $x_i$ or $x_j$ contains exactly $d$ newly 
added points and other segments on the path from $x_i$ to $x_j$ have $d-1$ points 
(shown using small circles in Figure. \ref{embedding}(c)), and 
(iv) only consecutive points on the path $x_i \rightsquigarrow x_j$ are within unit distance 
apart. 

Now, we construct a UDG $G''=(V'',E'')$, where  $V'' = X \cup Y$, and  
$E'' = \{(p_i,p_j) \mid p_i,p_j \in V'' \text{ and } d(p_i,p_j) \leq 1\}$. 
Here $|V''| = |X| + |Y| = n + \ell d$, and $|E''| = \ell d + m$, where $m$ 
is the number of edges in $G$. Thus, $G''$ can be constructed in polynomial time.
We will use the term {\it $d$-grid} for a grid whose each cell is of size $d \times d$.

The notion of points and vertices of $G''$ are used interchangeably in the rest of the paper.
Unless otherwise specified, the term distance refers to graph-distance.
\begin{lemma}\label{lemma_2}
Any D$d$IS of $G''$ contains at most $\ell$ points from $Y$.
\end{lemma} 
\begin{proof}
For each segment in the $d$-grid used to draw $G'$, the number of points of $Y$ appearing on it 
is $d$ or $d-1$. Thus, each segment may contain at most one point from $Y$ in the  D$d$IS of $G''$.
In particular, if two end-points of a segment $\eta$ of the $d$-grid (that are vertices of $G''$) are chosen 
in D$d$IS, then no point of $Y$ lying on $\eta$ will be chosen. Now, the result follows from the fact 
that $\ell$ many segments of the $d$-grid are used to draw $G'$.  
\end{proof} 
 
\begin{lemma}\label{lem:claim}
$G$ has a D$d$IS of cardinality at least $k$ if and only 
if $G''$ has a D$d$IS of cardinality at least $k+\ell$.
\end{lemma}

\begin{proof}
{(\bf Necessity)} Let $G$ have a D$d$IS $D$ of size at least $k$. 
Let $X' = \{x_i \in X \mid v_i \in D\}$. Let $G_{i,\alpha}$ denote a spanning tree of $G$ with the set of vertices 
$V_{i,\alpha} = \{v_j\in V(G)\mid d_G(v_i,v_j)\leq \alpha\}$. For each $v_i \in D$ start traversing from $x_i$ in $G''$.
Let $Y_i = \{y_\theta \in Y \mid d_{G''}(x_i,y_\theta)=d.\theta,\forall  \theta = 1,2,\ldots, \ell'\}$,
where $\ell'$ is the number of segments between $x_i$ and $x_j$, where $x_j$ corresponds to $v_j \in V_{i,\lfloor\frac{d}{2}\rfloor}$. 
Let $Y' = \bigcup_{x_i \in X'} Y_i$. The set $X'\cup Y'$ is a D$d$IS in $G''$.
Observe that there are some segments (corresponding to the edges which are not part 
of any $G_{i,\lfloor\frac{d}{2}\rfloor}$) that have not been traversed in the above process.
Now, we consider every such segment and choose the $\lceil\frac{d}{2}\rceil$-th point on it. Let $Y''$ be the set of chosen points. Needless to say, $Y''$ is also a D$d$IS of $G''$. 

By the way, we obtained the sets $Y'$ and $Y''$, there exists no pair of points 
$y_\alpha \in Y'$ and $y_\beta \in Y''$ such that $d_{G''}(y_\alpha,y_\beta) < d$. On the contrary, suppose $d_{G''}(y_\alpha,y_\beta) < d$. Implies, 
$y_\alpha$ and $y_\beta$ are from two segments, each having one, incident at some $x_j \in X\setminus X'$,
where $x_j$ corresponds to a leaf $v_j \in G_{i,\lfloor\frac{d}{2}\rfloor}$.
Note that $d_{G''}(y_\alpha,x_j) \geq \lfloor\frac{d}{2}\rfloor$ and $d_{G''}(x_j,y_\beta) \geq \lceil\frac{d}{2}\rceil$. Implies, $d_{G''}(y_\alpha,y_\beta) \geq d$, arrived at a contradiction.
Let $D'=X'\cup Y'\cup Y''$. As per our selection method each segment contributes
 one point in $Y'\cup Y''$. Thus, $|D'|\geq k+\ell$ since
 $|X'|\geq k$ and $|Y'\cup Y''|=\ell$.

\noindent
{(\bf Sufficiency)}
Let $G''$ have a D$d$IS $D'$ of cardinality at least $k+\ell$ and  
$D = \{v_i \in V \mid x_i \in D' \cap X\}$. Observe that $|D' \cap Y| \leq \ell$ 
(due to Lemma \ref{lemma_2}); so $|D|\geq k$. 
We shall show that, by suitably modifying $D$ (i.e., by removing or changing some of
the vertices in $D$), we get at least $k$ points from $X$ such that the set of
corresponding vertices in $G$ is a D$d$IS of $G$.
Consider a pair of vertices $v_i,v_j \in D$ such that  
$d_G(v_i,v_j)=d'  < d$ in $G$ (if there is no such pair, then $D$ is a D$d$IS of $G$ with $|D| \geq k$). 
Let $x_i,x_j\in D'\cap X$ be the vertices in $G''$ corresponding to $v_i,v_j \in D$, 
respectively. Also, let $\hat\ell$ be the number of segments on the path ${x_i \rightsquigarrow x_j}$ corresponding to the shortest path ${v_i \rightsquigarrow v_j}$. 
As each segment can contribute at most one point (from $Y$) in any solution, 
$D'$ can contain at most $\hat\ell+1$ points (including $x_i$ and $x_j$) from the path ${x_i \rightsquigarrow x_j}$. 

As per our construction of $G''$, the distance between $x_i$ and $x_j$ is $\hat\ell.d + 1$ in $G''$.
We update the solution along the  path ${x_i \rightsquigarrow x_j}$ as follows:
delete $x_j$ and other points of the path from $D'$.
Start traversing the path from $x_i$ and add every ${(d.\theta)}-$th point to $D'$, where $1 \leq \theta \leq \hat\ell$. 
The last point chosen is the point which is $d'$ distance away from $x_j$. 
The number of points in $D'$ on the path ${x_i \rightsquigarrow x_j}$ is $\hat\ell+1$.
Thus, $D'$ a new feasible solution in $G''$ whose size is at least as that of the 
previous solution. Observe that, the points in $D'$ that are on the segments outside the 
path ${x_i \rightsquigarrow x_j}$ will not be effected 
by the newly chosen points, and $|D'|\geq k+\ell$. 

We repeat the same for all pair of points in $D$ 
for which the shortest path distance in $G$ is less than $d$. Therefore,   
$|D|\geq k$ (from Lemma \ref{lemma_2}) and $D$ is a distance-$d$ 
independent set in $G$. 
\end{proof}

\begin{theorem}
GD$d$IS problem is NP-complete for unit disk graphs.
\end{theorem}
\begin{proof}
Follows from Lemma \ref{lem:np} and Lemma \ref{lem:claim}.
\end{proof}
\section{Approximation Algorithm}\label{sec:aa}

In this section, we discuss a simple 4-factor approximation algorithm for 
the GD$d$IS problem, for a fixed constant $d \geq 3$.
Let ${\cal R}$ be the rectangular region containing the point set $P$ (disk centers). 
From now on we deal with the point set $P$ rather than the UDG $G$ defined on $P$. 
We partition ${\cal R}$ into disjoint horizontal strips $H_1,H_2,\ldots,H_\nu$, 
each of width $d$ ($H_\nu$ may be of width less than $d$). The basic idea behind 
our algorithm is as follows:
\begin{itemize}
\item[(i)] Compute a feasible solution for each non-empty strip $H_i~(1 \leq i \leq \nu)$ 
independently as stated below:
\begin{itemize}
\item[] We split the horizontal strip into squares of size $d \times d$. In each square, 
we compute an optimum solution of the GD$d$IS problem defined by the points set inside that square. We consider 
all odd numbered squares and compute  the union $S^i_{odd}$ of  
optimum solutions of these squares. Similarly, the union $S^i_{even}$ of optimum solutions 
of all even numbered squares are also computed. Each of these is a feasible solution of 
GD$d$IS problem in $H_i$ as the minimum distance 
between each pair of considered squares is at least $d$. We choose $S^{i}=S^i_{even}$ or $S^i_{odd}$ 
such that $|S^i| = \max(|S^i_{even}|,|S^i_{odd}|)$ as the desired feasible solution for the strip $H_i$. 
\end{itemize} 

\begin{figure}[t]
\centering
\includegraphics[scale=0.9]{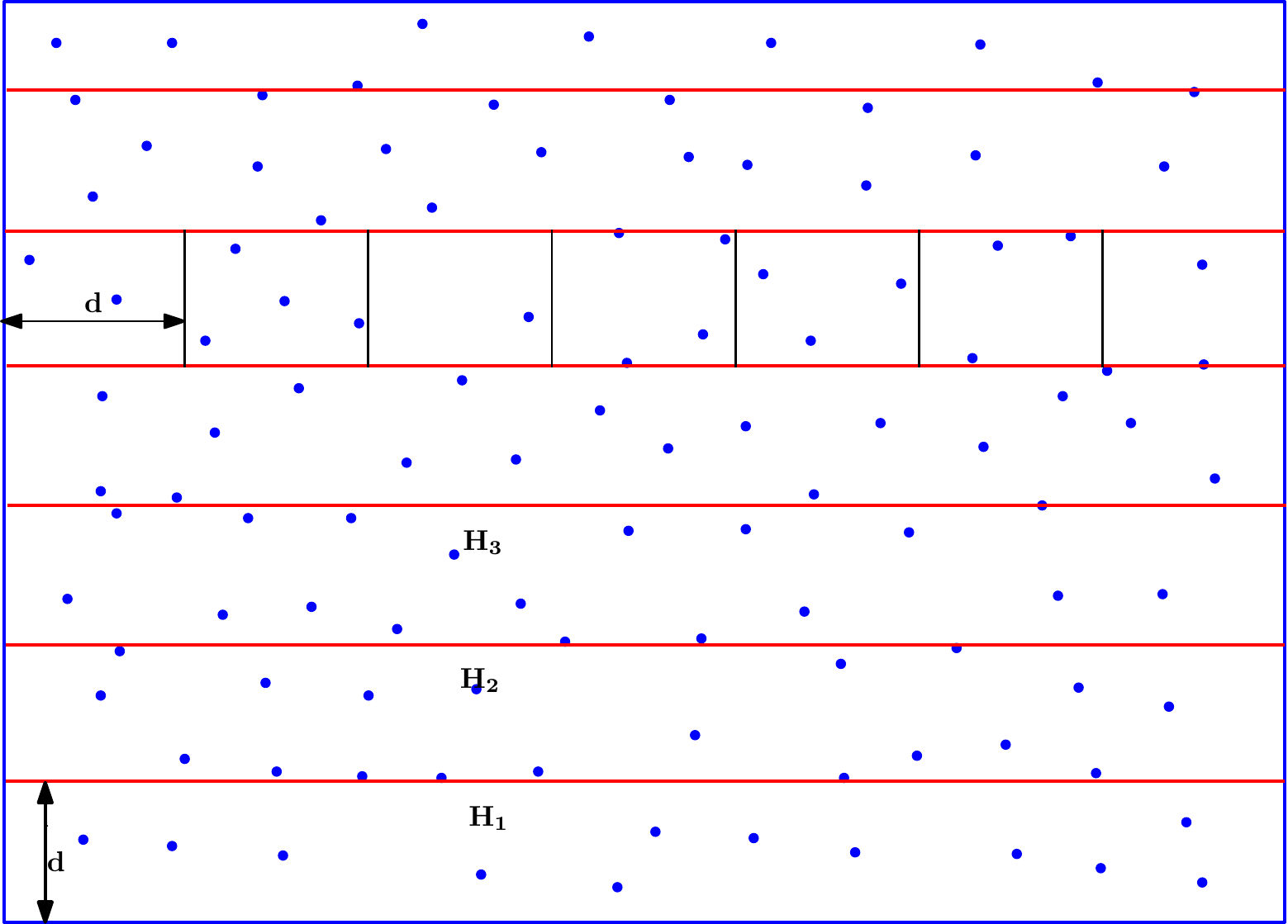}
\caption{A horizontal and vertical partition of the strips of width $d$} \label{fig:apx}
\end{figure}

\item[(ii)] Compute $S_{even}$ and $S_{odd}$, which are the union of the solutions of 
even and odd strips respectively, and 
\item[(iii)] Report $S^* = S_{even}$ or $S_{odd}$ such that 
$|S^*| = \max(|S_{even}|, |S_{odd}|)$ as a solution to the GD$d$IS problem. 
\end{itemize}
Note that, thus, the solution 
obtained in the above process is a feasible solution for the entire problem. 

\begin{lemma}
If $OPT$ is an optimum solution for the GD$d$IS problem, then 
$\max(|S_{even}|, |S_{odd}|) \geq \frac{1}{4}|OPT|$. 
\end{lemma}

\begin{proof}
Let us denote by $OPT^i$ an optimum solution of the non-empty strip $H_i$.
Since any two even (resp. odd) numbered strips, say $H_i$ and $H_j$, 
are at least $d$ distance apart, the feasible solutions computed in 
any method for $H_i$ and $H_j$ are independent\footnote{by independent 
we mean for any $p_i \in H_i \cap P$ and $p_j \in H_j \cap P$, $p_i$ and $p_j$ are 
distance-$d$ independent and also, $OPT^i \cap OPT^j = \emptyset$}. 
Thus, both $OPT_{even}=\bigcup\limits_{i \texttt{ is even}} OPT^i$ 
and $OPT_{odd}=\bigcup\limits_{i \texttt{ is odd}} OPT^i$
are feasible solutions for the given GD$d$IS problem. 

Note that $|OPT| \leq |OPT_{even}|+|OPT_{odd}| \leq 2 |OPT_*|$, where 
$OPT_*=OPT_{even}$ if $|OPT_{even}| > |OPT_{odd}|$; otherwise $OPT_*=OPT_{odd}$.

Also, note that we have not computed $OPT^i$ for the strip $H_i$. Instead, 
we have computed $S^i_{even}$ and $S^i_{odd}$ by splitting the strip $H_i$ 
into $d\times d$ squares, and accumulating the  optimum solutions of even and 
odd numbered squares separately. By the same argument as stated above, 
we have $|OPT^i| \leq 2|S^{i_*}|$, where $S^{i_*}=S^i_{even}$ if $|S^i_{even}| \geq 
|S^i_{odd}|$; otherwise $S^{i_*}=S^i_{odd}$.

Combining both the inequalities, we have 
$|OPT| \leq 4\max(|S_{even}|, |S_{odd}|)$. 
\end{proof}

\subsection{Solving a $d \times d$ square optimally} \label{dd}
Let ${\cal Q} \subseteq P$ be the set of points inside a $d \times d$ square 
$\chi$, and $G_\chi$ be the UDG defined on ${\cal Q}$. Let $C_1,C_2,\ldots,C_l$ be 
the connected components of $G_\chi$. Without loss of generality we assume that any two 
components in $G_\chi$ are at least $d$ distance apart\footnote{if there 
are two components having distance less than $d$ in $G$, then  
we can view them as a single component} in $G$.

\begin{lemma} \label{lem:lem1}
The worst case number of different connected components in $G_\chi$ is $O(d^2)$.
\end{lemma}
\begin{proof}
Partition $\chi$ into $O(d^2)$ cells, each of size 
$\frac{\sqrt{3}}{2} \times \frac{1}{2}$. The result follows from the fact that 
the points lying inside each cell are mutually connected.
\end{proof}   

In order to have the worst case size of a D$d$IS 
in $G_\chi$, we need to have an idea about the worst case size of a  
D$d$IS in a connected component in $G_\chi$.

\begin{lemma}\label{lem:comp_size}
 Let $C$ be any component of $G_\chi$. The number of mutually 
 distance-$d$ independent points in $C$ is bounded by $O(d)$.
\end{lemma}
\begin{proof}
Consider the square region $\chi'$ of size $3d \times 3d$ whose each side 
is $d$ distance away from the corresponding side of $\chi$. Let $Q' 
\subseteq P$ be the subset of points in $\chi'$. Partition $\chi'$ into 
cells of size $\frac{1}{2\sqrt{2}} \times \frac{1}{2\sqrt{2}}$. Thus, the number of cells 
in $\chi'$ is $O(d^2)$, and in each cell the unit disks centered at the 
points inside that cell are mutually connected. Let a pair of points 
$p_i,p_j \in C$ which are distance-$d$ independent. The shortest path 
$p_i\rightsquigarrow p_j$ between $p_i$ and $p_j$ entirely lies inside $\chi'$. If 
there is another point $p_k\in C$ which is distance-$d$ independent with 
both $p_i$ and $p_j$, then $p_k$ is at least distance $\frac{d}{2}$ away 
from each point on the path $p_i \rightsquigarrow p_j$. 
Thus, the path from $p_k$ to any point on the path $p_i \rightsquigarrow p_j$  
occupies at least $O(d)$ cells, and none of the points from these cells 
are distance-$d$ independent to all the points $p_i,p_j,p_k$. 
Thus, the addition of each point in the set of mutually distance-$d$ 
independent points in $\chi$ prohibits points in $O(d)$ cells to belong in 
that set, and hence the lemma follows.
\end{proof}

\begin{lemma}\label{lem:run_time}
 An optimal (i.e., maximum size) D$d$IS in $\chi$ can be computed in $d^2n^{O(d)}$ time.
\end{lemma}
\begin{proof}
We first construct a weighted complete graph $G'=(V',E')$ where $V'$ 
corresponds to the points in $Q'$. For each edge $(v_i,v_j)\in E'$, the 
weight $w(v_i,v_j)=1$ if $d(p_i,p_j) \leq 1$; otherwise $w(v_i,v_j)=\infty$.
 Next, we compute the 
all pair shortest paths between every pair of vertices in $G'$, and store them in a matrix $M$.  

By definition, intersection of distance-$d$ 
 independent sets of any two components is empty. Thus, a D$d$IS 
 of maximum size in $G_\chi$ can be computed by considering 
all  components of the UDG $G_\chi$, and computing 
the union of the D$d$IS of maximum sizes of those 
 components. We consider each  component of  
$G_\chi$ separately. For each component $C$, we consider all possible 
tuples of size at most $O(d)$ (due to Lemma \ref{lem:comp_size})  and for each 
tuple, we check whether they form a D$d$IS or not by 
consulting the matrix $M$ in $O(d^2)$ time. Thus, a maximum size 
D$d$IS in $C$ can be computed in $O(d^2|C|^{O(d)})$ 
time and the total time for computing a maximum size D$d$IS
in $G_\chi$ is $O(d^2 \sum\limits_{C \in G_{\chi}} |C|^{O(d)})$ = 
$d^2 n_\chi^{O(d)}$, where $n_\chi = \sum\limits_{C \in G_{\chi}} |C|$, the number of vertices in $G_\chi$.
\end{proof}
%
%


\begin{theorem}
 Given a set $P$ of $n$ points in the plane, 
 we can always compute a D$d$IS of size at least 
 $\frac{1}{4}|OPT|$ in $d^2n^{O(d)}$ time, where $|OPT|$ is the maximum cardinality of a GD$d$IS.
\end{theorem}
\begin{proof}
Follows from Lemma \ref{lem:lem1}, Lemma \ref{lem:comp_size} and Lemma \ref{lem:run_time}.
\end{proof}
\section{Approximation Scheme}\label{sec:ptas}

In this section, using the shifting strategy \cite{hochbaum}, we propose a 
polynomial time approximation scheme (PTAS) for the 
GD$d$IS problem, for a given fixed constant $d \geq 3$.
Let ${\cal R}$ be an axis parallel rectangular region containing the point set $P$ 
(i.e., centers of the disks of the given UDG). We use two-level nested shifting strategy. 
The first level executes $k$ iterations, where $k \gg d$. 
The $i$-th iteration ($1 \leq i \leq k$) of the first level is as follows: 
\begin{itemize}
\item[$\bullet$] Assuming $\cal R$ is left-open, partition ${\cal R}$ into vertical strips such that (a) first strip 
is of width $i$, (b) every even strip is of width $d$, and (c) every odd 
strip, except the first strip, is of width $k$. 
\item[$\bullet$] Without loss 
of generality, assume that the points lying on the left boundary of a 
strip belong to  
the adjacent strip to its left (i.e., every strip is left open and right closed). 
\item[$\bullet$] Compute some desired feasible solutions for the odd strips (of 
width $k$). These solutions can be merged to produce a solution of the entire 
problem since these odd numbered strips are distance-$d$ apart. 
\end{itemize}

The second level of the nested shifting strategy is used to find a solution 
for an iteration in the first level. We consider each non-empty odd strip
separately, and execute $k$ iterations. In the $i$-th iteration, we partition 
it horizontally as in the first level (mentioned in the first bullet above). 
We get a solution of a strip by solving each $k \times k$ square in that 
strip optimally. The union of the
solutions of all the odd numbered squares/rectangles in that strip is the desired solution of that vertical strip of the first level. Finally, we take 
the union of the solutions of all the odd vertical strips to compute the 
solution of that iteration of the first level. Thus, we have the solutions 
of all the iterations of the first level. We report the one having the 
maximum cardinality as the solution of the given GD$d$IS problem. Compute a matrix $\cal M$ containing the cost of all pair shortest paths in a complete graph defined with the points in $P$ where 
the edge costs are as defined in Section \ref{dd}. The method 
of computing an optimum solution inside a $k \times k$ square is described below.

\begin{figure}[t]
\centering
\includegraphics[scale=0.9]{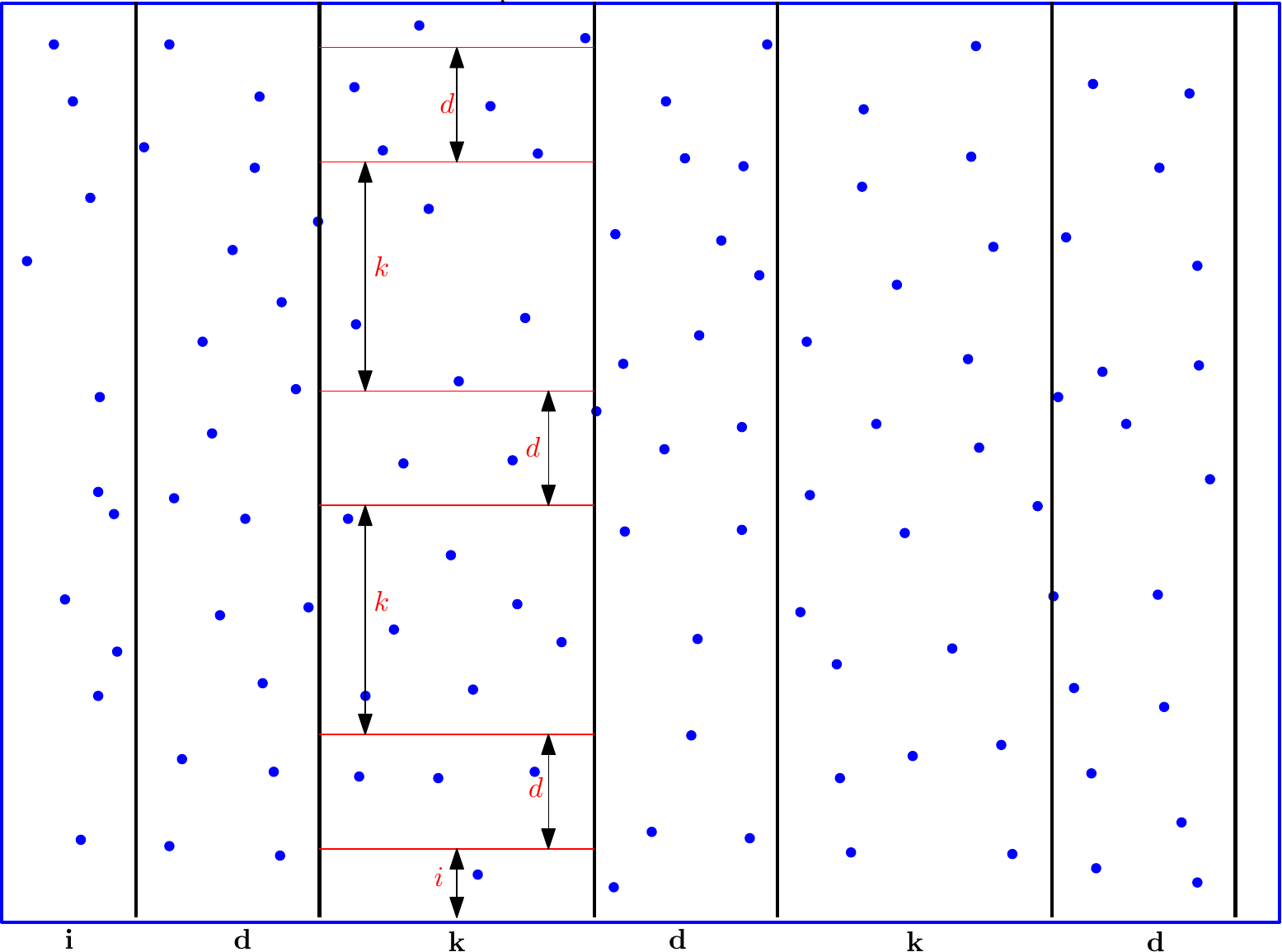}
\caption{Horizontal and vertical partition of the strips}\label{fig:ptas1}
\end{figure}

\subsection{Computing an optimum solution in a $k \times k$ square}\label{subsec:1}
We apply a divide and conquer strategy to compute an optimum solution 
of the GD$d$IS problem defined on a set of points ${\cal Q} \subseteq P$ 
inside a square $\chi$ of size $k \times k$. We partition $\chi$ into four sub-squares, each of size 
$\frac{k}{2}\times \frac{k}{2}$, using a horizontal line $\ell_h$ and a 
vertical lines $\ell_v$ (see Figure. \ref{fig:ptas}). Let ${\cal Q}_1 \subseteq {\cal Q}$ be the subset 
of points in $\chi$ which are at most $d$ distance away from  
$\ell_h$ and/or $\ell_v$. Let ${\cal Q}_2$ be a maximum cardinality 
subset of ${\cal Q}_1$ such that all the points in ${\cal Q}_2$ are pair 
wise distance-$d$ independent in $p$. 

\begin{figure}[t]
\centering
\includegraphics[scale=0.9]{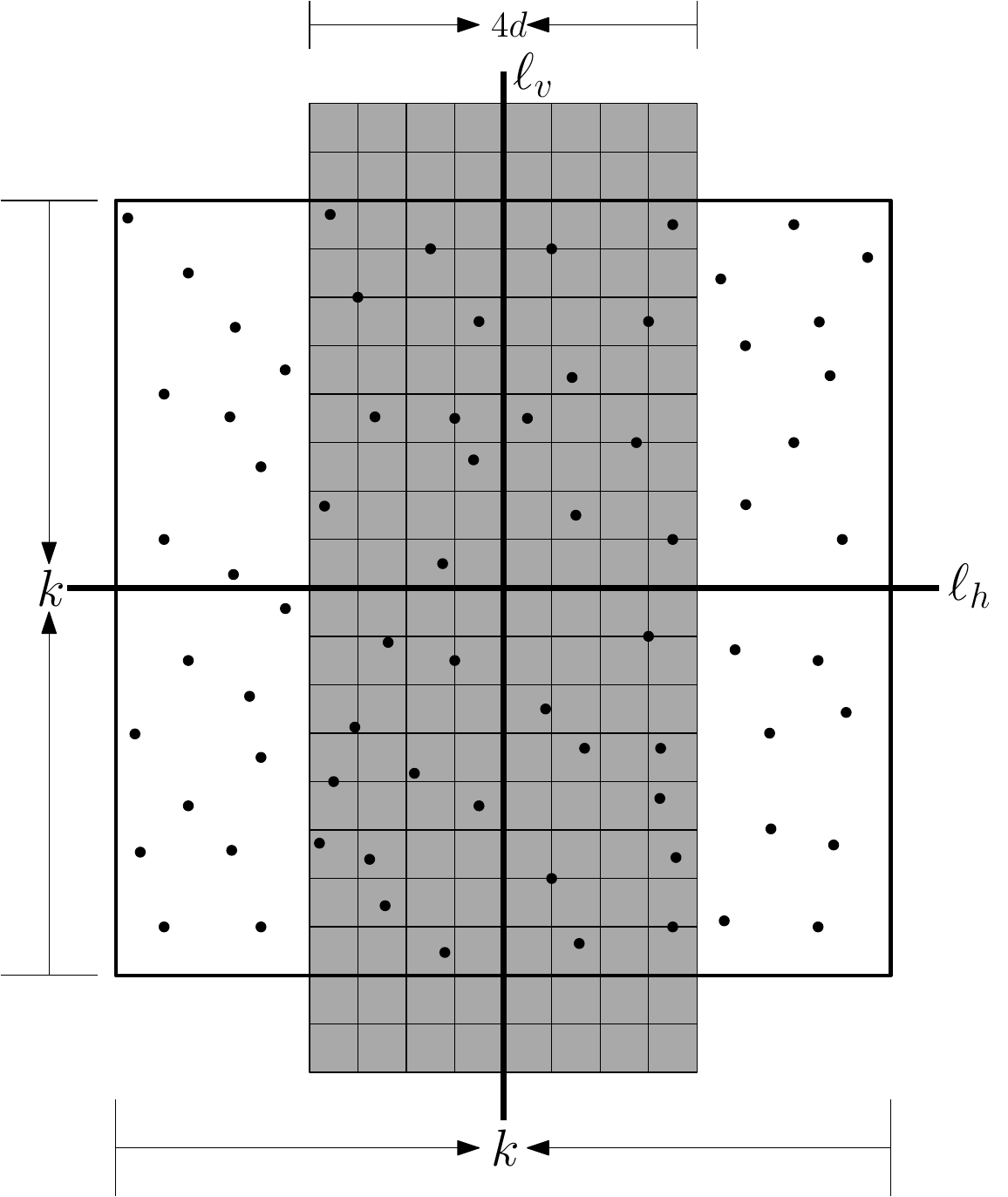}
\caption{A strip of width $4d$ around the vertical line $\ell_v$ 
(shown in dotted lines)}\label{fig:ptas}
\end{figure}
\begin{lemma}\label{lem:lemma_bound}
$|{\cal Q}_2|\leq O(k)$.
\end{lemma}
\begin{proof}

The proof follows from the 
 similar combinatorial argument discussed in 
the proof of Lemma \ref{lem:comp_size}.
\end{proof}

We apply the divide and conquer strategy on $\chi$ as follows:

\begin{description}
\item[Step 1:] Choose all possible subsets of points 
of sizes at most $O(k)$ in ${\cal Q}_1$. 
\item[Step 2:] For each subset, do the following:
\begin{itemize}
\item[$\bullet$] Check whether they are mutually distance-$d$ independent by 
consulting the table $\cal M$. If so, then they form ${\cal Q}_2$.
\item[$\bullet$] Consult the matrix $\cal M$ to delete the points in 
$\chi$ which are at most distance $d-1$ away from each member in ${\cal Q}_2$. 
\item[$\bullet$] Recursively solve the four independent sub-problems defined
by the points of ${\cal Q}\setminus {\cal Q}_1$ in the four quadrants 
$\chi_1$, $\chi_2$, $\chi_3$, $\chi_4$ 
defined by $\ell_h$ and $\ell_v$. 
\item[$\bullet$] Return ${\cal Q}_2$ = ${\cal Q}_2\cup (\bigcup\limits_{i=1}^4 {\cal Q}^i_2$), 
where ${\cal Q}^i_2$ is the solution of the sub-problem on the points of 
$\chi_i$. 
\item[$\bullet$] Retain the solution for the present subset if it is better 
than the solutions produced by earlier choices of ${\cal Q}_2$. 
\end{itemize}

\end{description}
  
 \begin{lemma}\label{lem:apprx_tc}
 The solution produced for the cell $\chi$ (of size $k \times k$) in the 
 aforesaid process is optimum, and the time complexity of the proposed 
 algorithm is $k^2m^{O(k)}$, where $m=|{\cal Q}|$.
 \end{lemma}
 
\begin{proof} 
Let $OPT_\chi$ be an optimal solution for the points lying in $\chi$.
Note that our process checks all combinations of points of size $|OPT_\chi|$.
Thus, the combination of points in $OPT_\chi$ must appear at some stage in the process.

If $T(m,k)$ denote the time complexity of computing the distance-$d$ 
independent set in $\chi$, then $ T(m,k) = 4\times T(m,\frac{k}{2})\times m^{O(k)} + O(k^2)$, 
which is $k^2 \times m^{O(k)}$ in the worst case.
\end{proof}

Using the analysis of \cite{hochbaum}, we have the following result.

\begin{theorem}
 Given a set $P$ of $n$ points (centers of the unit disks) in the plane and an integer $k >1$, 
 the proposed scheme produces a D$d$IS of size at least 
 $\frac{1}{(1+\frac{1}{k})^2}|OPT|$ in $k^2n^{O(k)}$ time, where $|OPT|$ is the maximum cardinality of a GD$d$IS.
 \end{theorem}
\section{The Hardness Result of GD$d$DS Problem on Unit Disk Graphs} \label{sec:hardness-1}
The decision version of the GD$d$DS problem, denoted by D(GD$d$DS), for a fixed integer $d\geq 2$, is defined as follows:
\begin{description}
 \item[Input.] An unweighted unit disk graph $G=(V,E)$ defined on a point set $P$ and a positive integer $k \leq |V|$.
 \item[Question.] Does there exist a distance-$d$ dominating set of size at most $k$ in $G$?
\end{description}
\begin{lemma}\label{lem:np}
D(GD$d$DS) problem is NP.
\end{lemma}
\begin{proof}

For a given set of vertices, we can verify whether all the vertices of the input graph are distance-$d$ dominated or not in polynomial time using Floyd-Warshall's all-pair shortest path algorithm \cite{cormen}. Hence D(GD$d$DS) $\in$ NP. 
\end{proof}

Now, for proving the problem belongs to NP-hard class, we do polynomial time reduction from a known NP-hard problem, the {\it vertex cover} problem \cite{garey2002} defined on planar graphs with maximum degree 3, to it. The decision version of the vertex cover problem defined on planar graphs with maximum degree 3, denoted by D($VC_p$), is defined as follows:
\begin{description}
 \item [The vertex cover problem on planar graphs] (\textsc{$VC_p$})
 \item [Input:] An undirected planar graph $G$ with maximum degree 3 and a positive integer $k$.
 \item [Question:] Does there exist a vertex cover $D$ of $G$ such that $|D|\leq k$?.
\end{description}
The following corollary follows from Lemma \ref{embedding}:
\begin{corollary}\label{cor:graph_embed-1}
 A planar graph $G=(V,E)$ with maximum degree 3 can be embedded on a plane having grid cell of size $2d \times 2d$, 
 so that its vertices lie at points of the form $(i * 2d,j * 2d)$ and its edges are drawn using a 
 sequence of consecutive line segments drawn on the vertical lines of the form $x=i * 2d$ and/or horizontal lines of the form $y=j * 2d$, for some integers $i$ and $j$ \emph{(see Figure \ref{fig:embedding-1})}.
\end{corollary}

\begin{lemma}\label{lemma-udg-1}
 Let $G=(V,E)$ be an instance of D(\textsc{$VC_p$}) with maximum degree 3. An instance $G'=(V',E')$ of D(GD$d$DS) can be constructed from $G$ in polynomial-time.
\end{lemma}
\begin{figure}[!ht]
  \centering
  \begin{minipage}{.45\textwidth}
  \hspace{-3.9cm}
  \centering
  \includegraphics[scale=0.99]{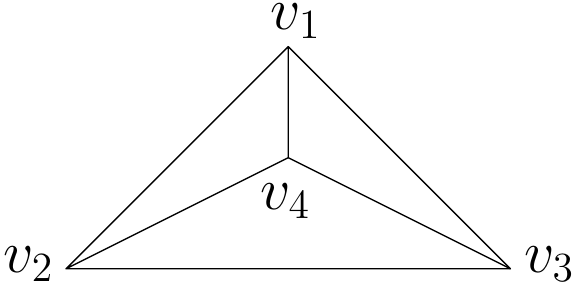} \\ \hspace{-3.5cm}{(a)}
  \end{minipage}
  \begin{minipage}{.45\textwidth}
  \hspace{-2.6cm}
  \centering
  \includegraphics[scale=0.70]{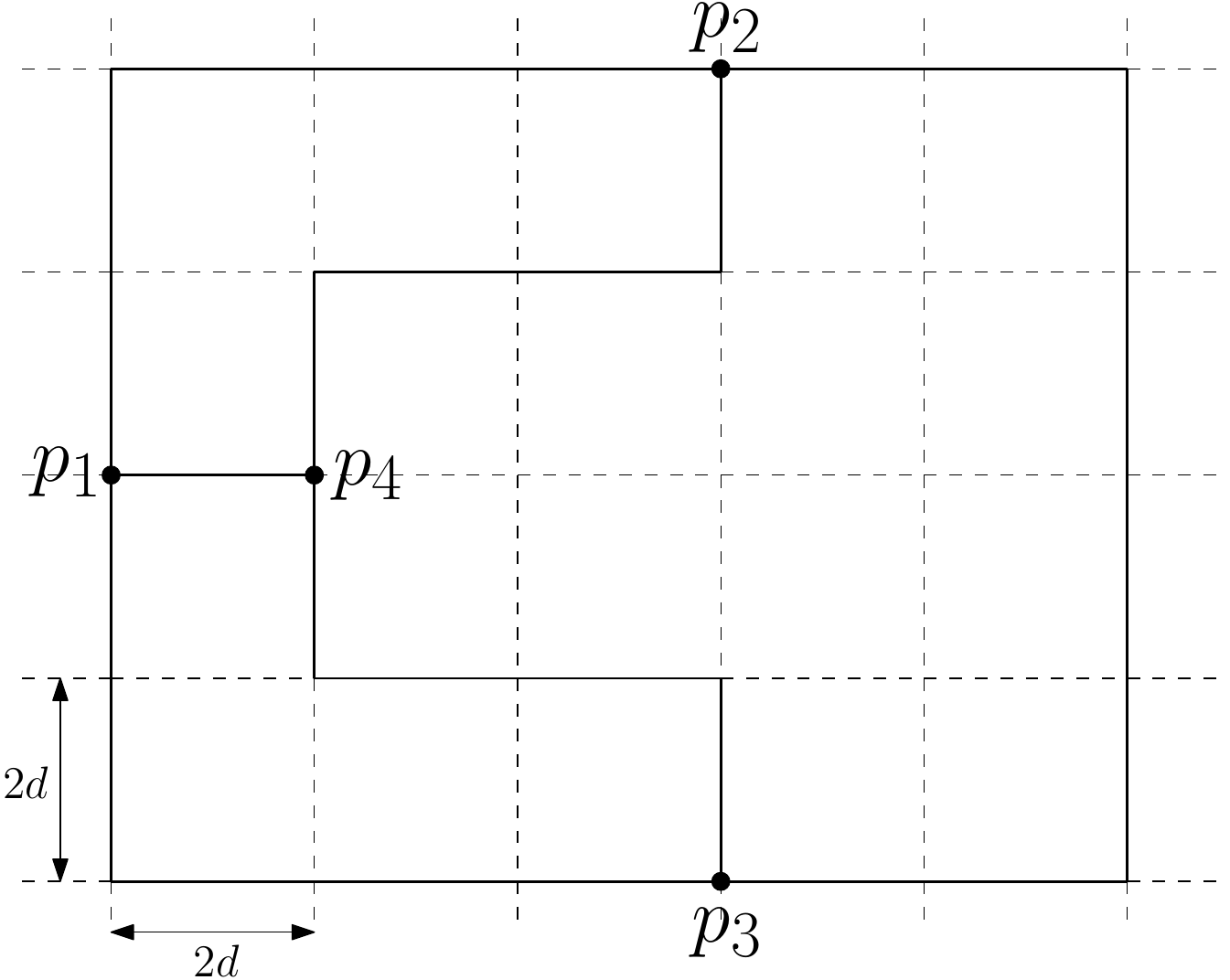}\\
  {(b)}
  \end{minipage}
  \begin{minipage}{.45\textwidth}
  \centering
  \hspace*{-1.7cm}
  \includegraphics[scale=0.7]{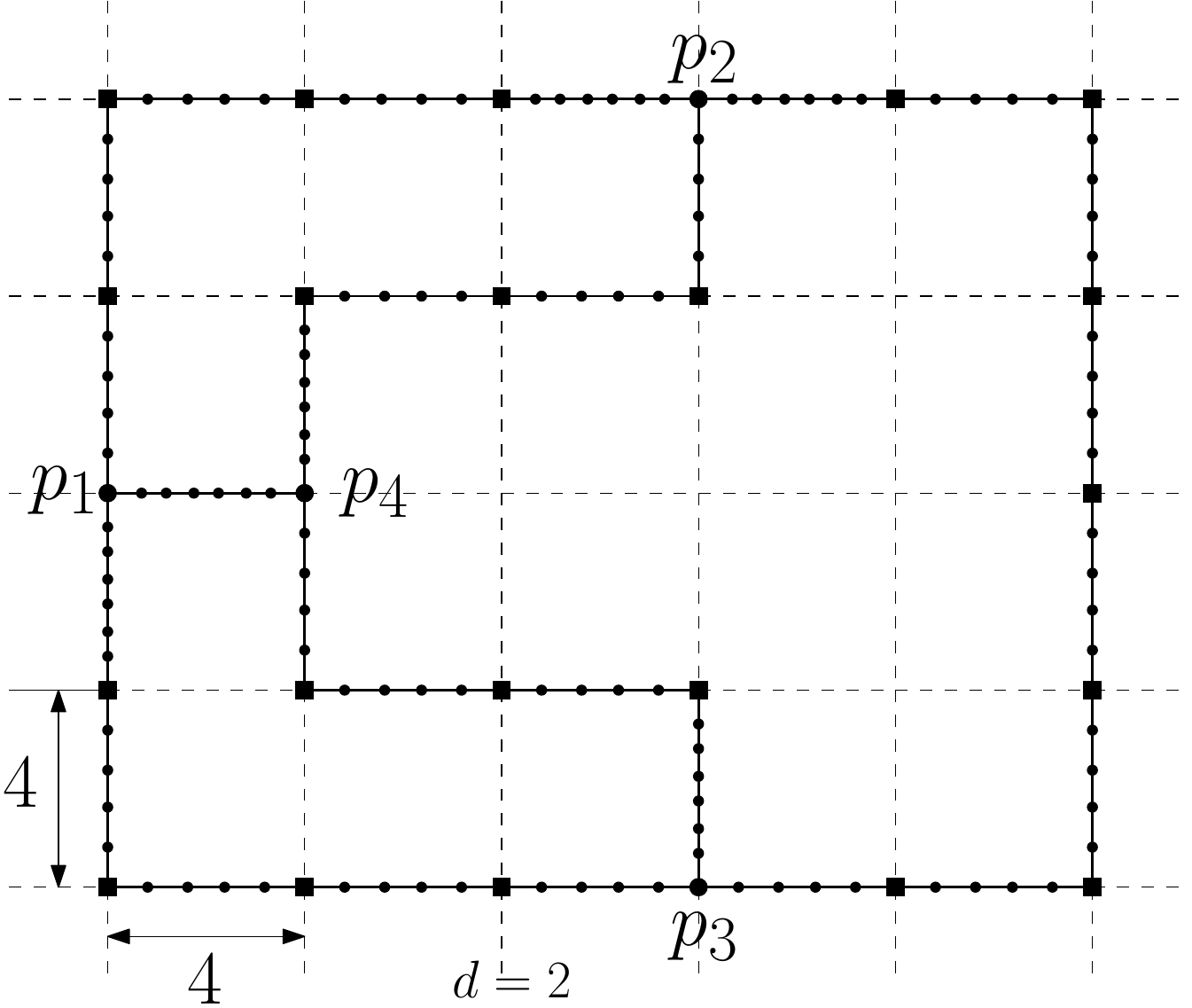}\\
  {(c)}
  \end{minipage}
  \begin{minipage}{.45\textwidth}
  \centering
  \hspace{1.5cm}
  \includegraphics[scale=0.65]{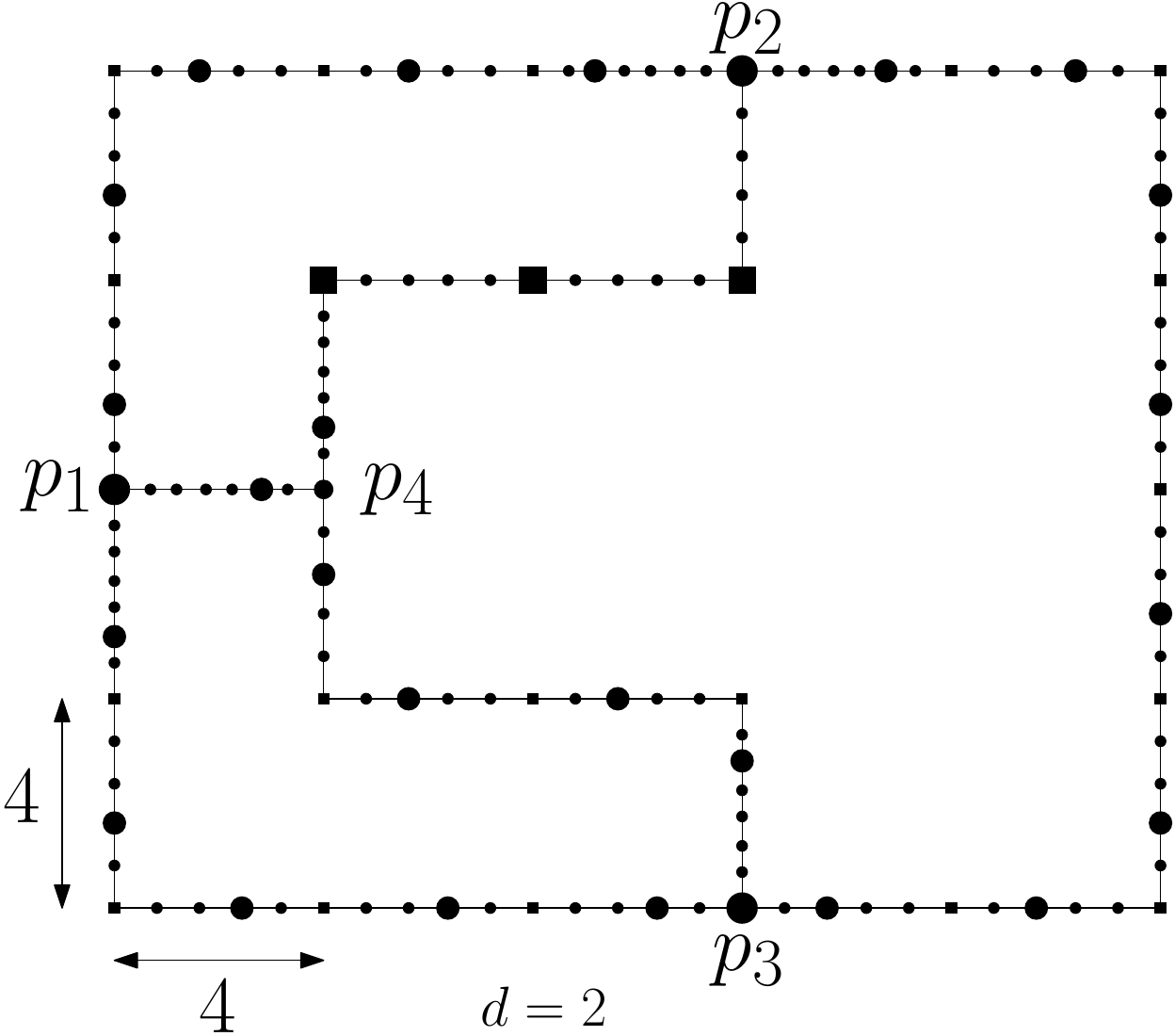}\\
  {(d)}
  \end{minipage}
  \caption{(a) A planar graph $G$ of maximum degree 3, (b) its embedding $G'$ on a grid of cell size
  $4\times 4$, (c) adding of extra points to $G'$, (d) the obtained UDG $G'$.}\label{fig:embedding-1}
  \end{figure}
\begin{proof}
Embed the instance $G$ on the plane as discussed in corollary \ref{cor:graph_embed-1},  using one of the algorithms in \cite{hopcroft,itai}. An edge in the embedding is a sequence of connected line segment(s) of length $2d$ units each. Let $\ell$ be the total number of line segments in the embedding. We add points $P=\{p_1,p_2,\ldots,p_n\}$ corresponding to the vertices $V=\{v_1,v_2,\ldots,v_n\}$ in the embedding. To make $G'$ a UDG we introduce a set $Q$ of 
extra points on the segments used to draw the edges of $G'$. Thus, the set of points in $P$ together with $Q$ form a UDG $G'$. Let $(p_i,p_j)$ be an edge in $G'$ 
corresponding to the edge $(v_i,v_j)$ in $G$ and has $\ell'$ grid segments. 

{\bf Case-1:} If $\ell'=1$, then we add $3d$ points $p_{ij}^1,p_{ij}^2,\ldots,p_{ij}^{3d}$ on the segment such that the Euclidean distance of $p_i$ to $p_{ij}^1$ and $p_{ij}^{3d}$ to $p_j$ is 0.72 and the Euclidean distance between $p_{ij}^t$ and $p_{ij}^{t+1}$ is $\frac{2d-1.44}{3d-1}>0.5$ for $t=1,2,3,\ldots,3d-1$. Therefore, the length of the path from $p_i$ to $p_j$ is exactly $3d+1$ (for $d=2$, see the edge ($p_1,p_4$) in Figure \ref{fig:embedding-1}(c)).\\
{\bf Case-2:} If $\ell'>1$, then we consider all joint points of each pair of consecutive segments other than the points of $P$ and add a point for each joint  points to the set $Q$ (see the square points in the edge ($p_1,p_3$) in Figure \ref{fig:embedding-1}(c)). Then, we add $3d$ points (see Case-1) in one of the two segments for which one end is associated either with $p_i$ or $p_j$ and $2d$ points $p_{ij}^1,p_{ij}^2,\ldots,p_{ij}^{2d}$ on the segment such that the Euclidean distance of $p_{ij}^1$ and $p_{ij}^{2d}$ from the end points of the segment is 0.75 and the Euclidean distance between $p_{ij}^t$ and $p_{ij}^{t+1}$ is $\frac{2d-1.5}{2d-1}>0.5$ for $t=1,2,3,\ldots,2d-1$ in the remaining $\ell'-1$ segments in such a way that  after adding the extra points, the length of one segment is $3d+1$ and the length of all other segments is $2d+1$ (see the edge ($p_2,p_3$) in Figure \ref{fig:embedding-1}(c)).\\
Note that in both the cases, only consecutive points on the path $p_i \rightsquigarrow p_j$ are within unit distance apart. Let $Q$ be the set of points generated in these two cases. 

Observe that $G'=(V',E')$ is a UDG, where  $V' = P \cup Q$, and $E' = \{(p_i,p_j) \mid p_i,p_j \in V' \text{ and } d(p_i,p_j) \leq 1\}$. Here $|V'| = |P| + |Q| \leq n + 3\ell d$, and $|E'| \leq 3\ell d + m$, where $m$ is the number of edges in $G$ and $\ell$ is bounded by $n$. Thus, $G'$ can be constructed in polynomial time.
 \end{proof}
\begin{lemma}\label{lem:claim}
  $G$ has a vertex cover of size at most $k$
 if and only if $G'$ has a distance-$d$ dominating set of size at most $k + \ell$.
 \end{lemma} 
\begin{proof}

({\bf Necessity}) Let there exist a vertex cover $D$ of size at most $k$ for the graph $G$. Let ${\cal S}$ be the collection of points from $P$ in $G'$ corresponding to the vertices of $D$ in $G$, i.e., ${\cal S}= \{p_i \in P \mid v_i \in D\}$. Note that $|{\cal S}|=|D|$. We choose one point from each segment in such a way (see next paragraph) that the selected points along with ${\cal S}$ form a distance-$d$ dominating set of $G'$ and the size of ${\cal S}$ is at most $k+\ell$.  

Note that, every edge in $G$ has at least one of its end vertices in $D$ ($D$ is a vertex cover in $G$). 
For each edge $(v_i,v_j)$ in $G$, start traversing from the corresponding vertex $p_i$ in $G'$ (if $v_i \in D$ or from $p_j$, if $v_j\in D$) in the embedding and select each $2d+1$-th vertex in ${\cal S}$ encountered from $p_i$ to $p_j$ in the traversal (see $(p_2,p_3)$ in Figure \ref{fig:embedding-1}(d). The big vertices are part of ${\cal S}$ while traversing from $p_2$). Observe that, ${\cal S}$ is a distance-$d$ dominating set in $G'$ having $|{\cal S}|\leq k+\ell$ as we have chosen one vertex from each segment in the embedding and the way we have choosen ${\cal S}$, for any $p_i\in V'$  there is always exist at least one point $p_j\in {\cal S}$ such that $d(p_i,p_j) \leq d$.

 \begin{figure}[!ht]
  \centering
  \begin{minipage}{.45\textwidth}
  \hspace{-1.3cm}
  \centering
  \includegraphics[scale=0.99]{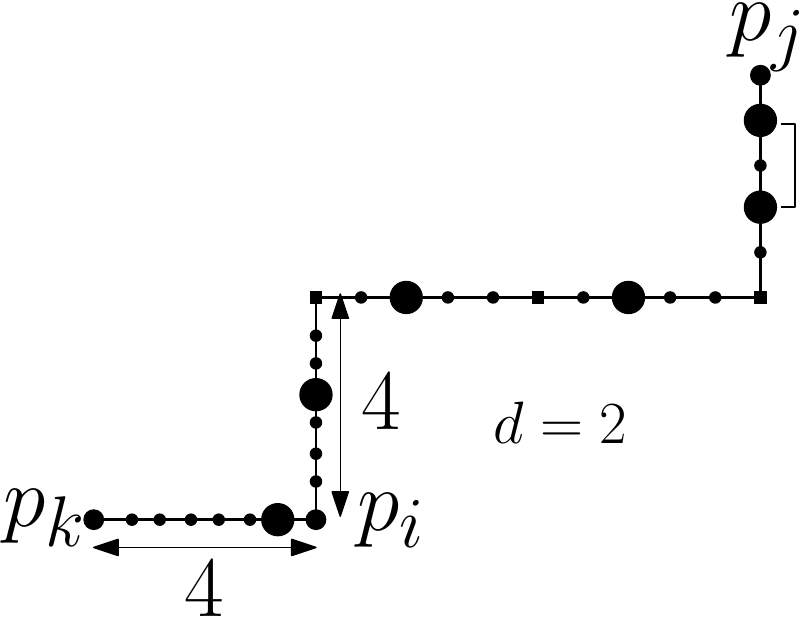} \\ \hspace{-3.5cm}{(a)}
  \end{minipage}
  \begin{minipage}{.45\textwidth}
  \centering
  \includegraphics[scale=0.70]{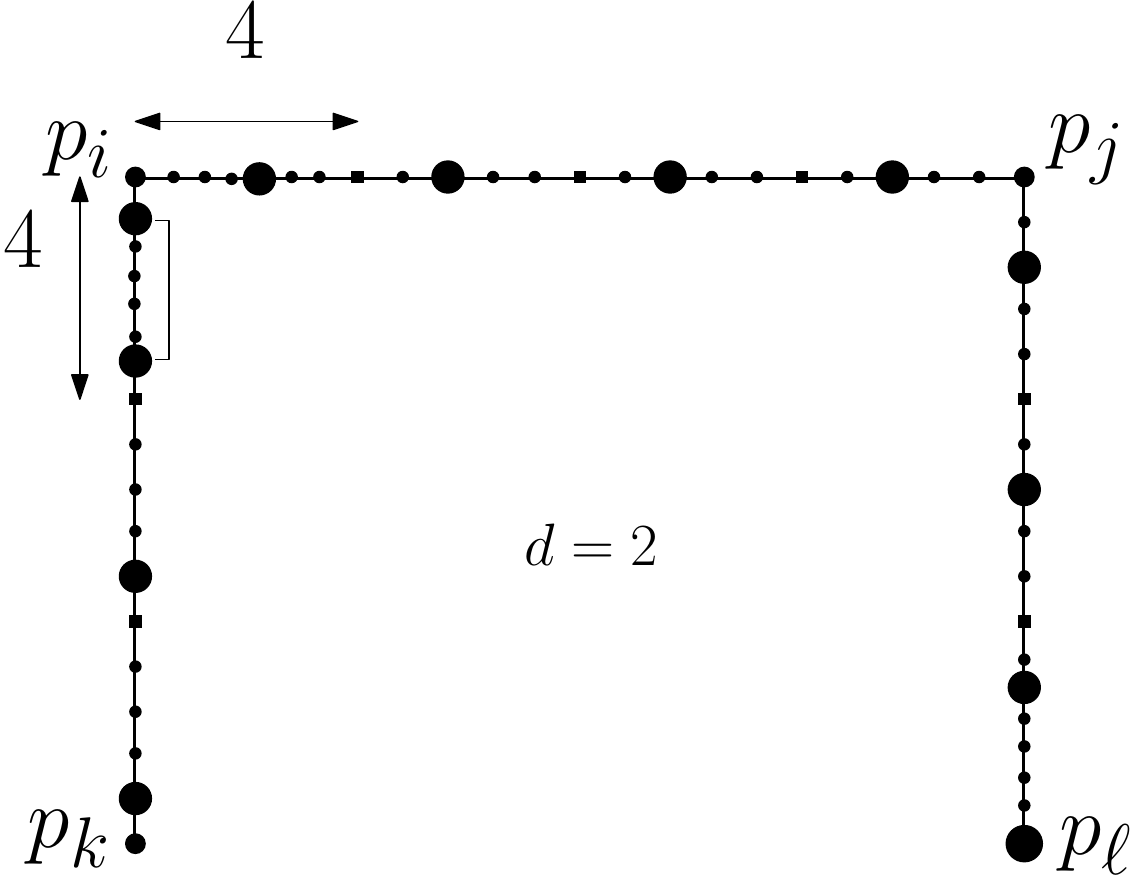}\\
  {(b)}
  \end{minipage}
    \caption{(a) $p_j$ is only connected with $p_i$, (b) $p_i$ connected with $p_k$ and $p_j$ connected with $p_{\ell}$.}\label{fig:suff}
  \end{figure}

 ({\bf Sufficiency}) Let ${\cal S} \subseteq V'$ be a GD$d$DS of size at most $k + \ell$ in $G'$. We need to prove that $G$ has a
 vertex cover of size at most $k$. 
Let $D = \{v_i \in V \mid p_i \in {\cal S} \cap P\}$. Observe that $|D|\leq k$ as the length of each segment in $G'$ is at least $2d+1$ there must be at least one point from each segment choosen in ${\cal S}$. It remains to prove that $D$ is a vertex cover of $G$.
 If any edge $(v_i,v_j)$ in $G$ has none of its end 
 vertices in $D$, then consider the points $p_i$ and $p_j$ corresponding to $v_i$ and $v_j$ respectively. \\
 {\bf Case (i)}: If $p_j$ is the only vertex that is connected with $p_i$ in $G'$, then the chain of segments (say $\ell'$) in the path $p_i \rightsquigarrow p_j$ in $G'$ has at least $\ell'+1$ vertices in ${\cal S}$ (see Figure \ref{fig:suff}(a) for example). In this case, we delete one point from the segment containing two points in ${\cal S}$ and introduce $p_i$ in ${\cal S}$.\\ 
 {\bf Case (ii)}: If both $p_i$ and $p_j$ are connected with some points $p_k$ and $p_{\ell}$ respectively in $G'$, then either the chain of segments (say $\ell'$) in the path $p_i \rightsquigarrow p_j$ in $G'$ has at least $\ell'+1$ vertices in ${\cal S}$ (see Case (i)) or the chain of segments (say $\ell'$) in the path $p_i \rightsquigarrow p_k$ or ( $p_j \rightsquigarrow p_{\ell}$) in $G'$ has at least $\ell'+1$ vertices in ${\cal S}$ (see Figure \ref{fig:suff}(b) for example). In this case, we choose the segment having two points in ${\cal S}$ and remove one point of the segment from ${\cal S}$ and introduce $p_j$ in ${\cal S}$ if $p_k \in {\cal S}$ otherwise introduce $p_i$ in ${\cal S}$. Update $D$ and repeat the process till every edge has at least one of its end vertices in $D$. Note that, in both the cases, we delete at most one point from such segments having two of its points in the solution and there does not exist a segment in $G'$ having none of its points in ${\cal S}$, which leads the proof that  $D$ is a vertex cover in $G$ with $|D|\leq k$.
 \end{proof}
 \begin{lemma}\label{lem:nphard}
 D(GD$d$DS) problem is NP-hard.
 \end{lemma}
 \begin{proof}
 Follows from Lemma \ref{lemma-udg-1} and Lemma \ref{lem:claim}.
 \end{proof}
 \begin{theorem}
 D(GD$d$DS) problem is NP-complete.
 \end{theorem}
 \begin{proof}
 From Lemma \ref{lem:np}, D(GD$d$DS) $\in$ NP and from Lemma \ref{lem:nphard}, D(GD$d$DS) $\in$ NP-hard. Therefore, D(GD$d$DS) $\in$ NP-complete.
\end{proof}
\section{Approximation Algorithm for GD$d$DS Problem} \label{sec:apprx}
In this section, we explain a $4$-factor approximation algorithm for distance-$d$ dominating set problem.
Let ${\cal R}$ be the smallest rectangular region containing the point set $P$ (disk centers).  We partition ${\cal R}$ into squares having side length $\frac{3}{\sqrt{2}}d \times \frac{3}{\sqrt{2}}d$ (see Figure \ref{fig:dom_approx}(a)).
The basic idea behind our algorithm is as follows:
\begin{itemize}
\item color the partitioning squares with 4-colors such that the distance between two same colored squares are more than $2d$ (see Figure \ref{fig:dom_approx}(a)).
\item find the optimal solution of each squares (see Subsection \ref{subsec:opt}).
\item let $OPT_i$ denotes the optimal solution generated by our algorithm for the squares having color $i$, for $i=1,2,3,4$.
\item Let $OPT$ be the minimum distance-$d$ dominating set of the graph. Therefore, $|OPT_i| \leq |OPT|$. Thus $\sum_i|OPT_i|\leq 4*|OPT|$.
\end{itemize}

\begin{figure}[!ht]
  \centering
  \begin{minipage}{.45\textwidth}
  \centering
  \hspace*{-1cm}
  \includegraphics[scale=0.65]{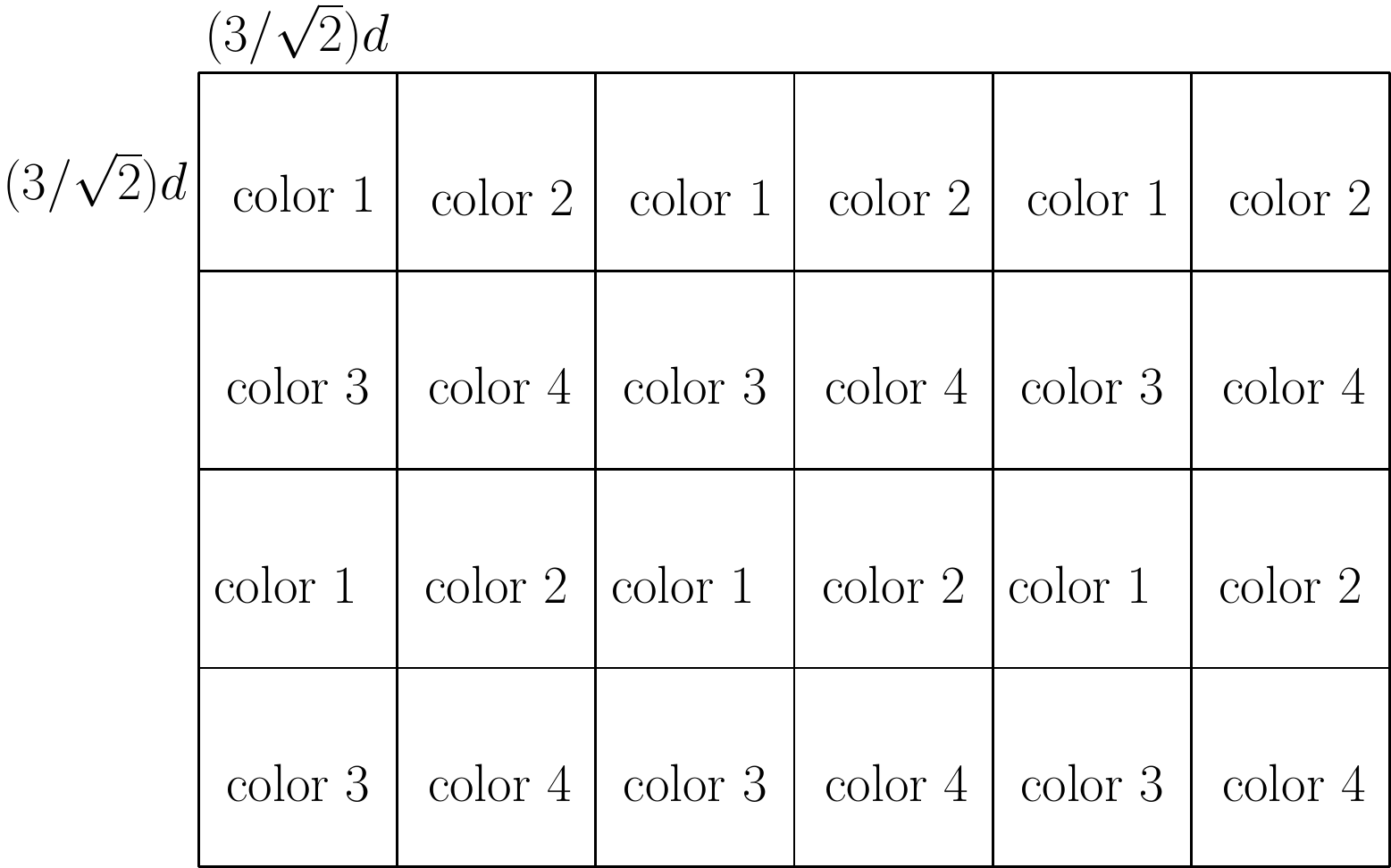} \\ {(a)}
  \end{minipage}
  \begin{minipage}{.45\textwidth}
  \centering
  \hspace*{2.2cm}
  \includegraphics[scale=0.55]{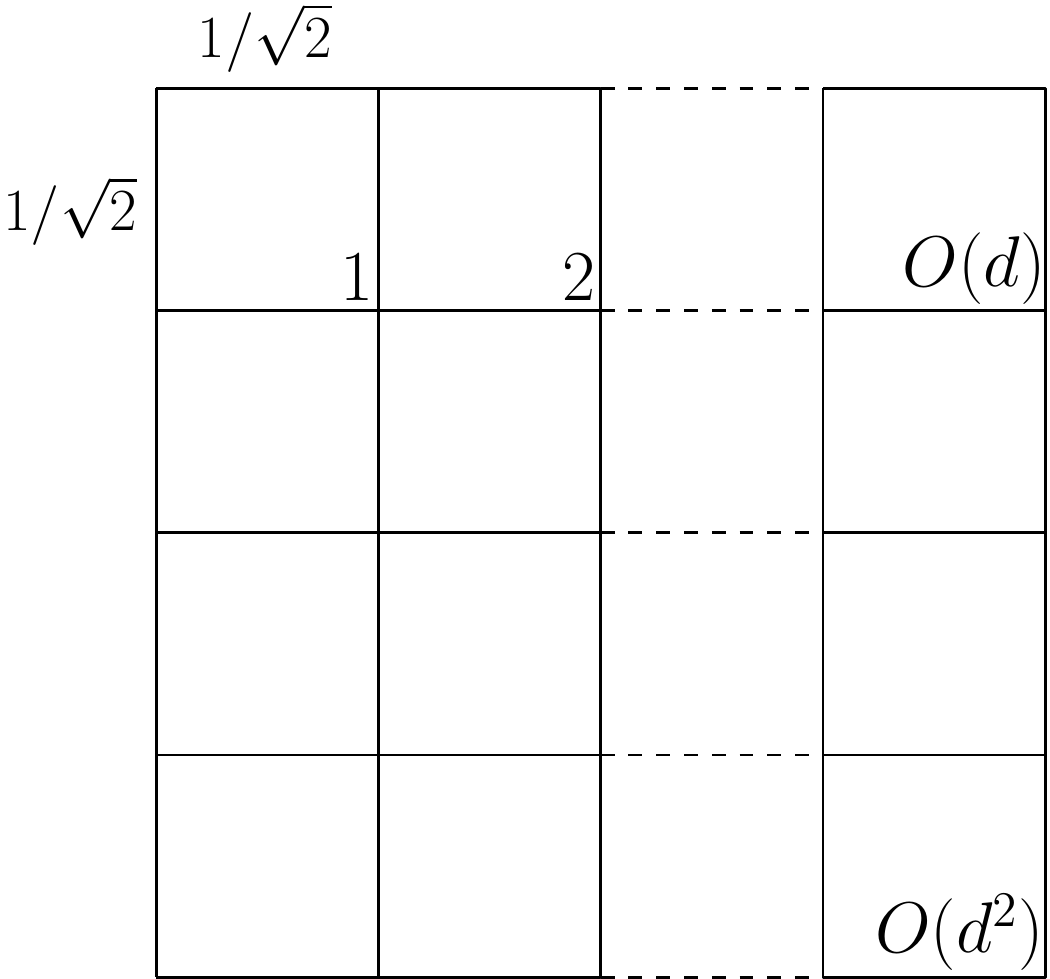}\\
  {(b)}
  \end{minipage}
  \begin{minipage}{.45\textwidth}
  \centering
  \hspace{-2.0cm}
  \includegraphics[scale=0.85]{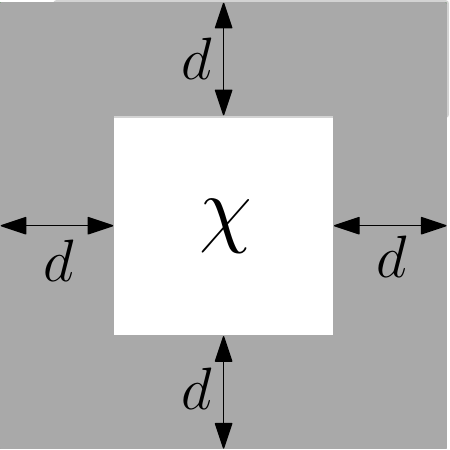}\\
  {(c)}
  \end{minipage}
    \caption{(a) partition of \emph{R} into smaller cells of size $\frac{3}{\sqrt{2}}d \times \frac{3}{\sqrt{2}}d$, (b) one cell partitioned into $O(d^2)$ sub-cells, and (c) one cell surrounded with $d$ width region.}\label{fig:dom_approx}
  \end{figure}

\subsection{Computing a minimum D$d$DS in a $\frac{3}{\sqrt{2}}d \times \frac{3}{\sqrt{2}}d$ square} \label{subsec:opt}
Let $\chi$ be a single $\frac{3}{\sqrt{2}}d \times \frac{3}{\sqrt{2}}d$ cell with $P_\chi \subseteq P$ be the set of points inside $\chi$, and $G_\chi$ be the UDG defined on $P_\chi$. Let $C_1,C_2,\ldots ,C_\ell$ be the different connected components of $G_\chi$ with the constraint that each $C_i$ ($1\leq i\leq \ell$) are $d$-distance apart from each other in $G$ (i.e., the UDG corresponding to points set $P$). If the distance between any two components is less than $d$, then combine these two components as a single component.

\begin{lemma} \label{lem:1}
The worst case number of different connected components in $G_\chi$ is $O(d^2)$.
\end{lemma}
\begin{proof}
 Partition $\chi$ into sub-cells of size $\frac{1}{\sqrt{2}}\times \frac{1}{\sqrt{2}}$ (see Figure \ref{fig:dom_approx}(b)). Therefore, the total number of sub-cells is $O(d^2)$. Every pair of points  within a sub-cell are connected as they are at most unit distance apart. Therefore, the points lying inside each sub-cell are in a same connected component. Thus, the lemma follows. 
\end{proof}


\begin{lemma} \label{lem:apprx_cmpnt}
The size of MD$d$DS in any connected component $C \in \{C_1,C_2,\ldots,C_\ell\}$ of $G_\chi$ is bounded by $O(d)$.
\end{lemma}
 
\begin{proof}
The proof follows from Lemma \ref{lem:comp_size}, with the fact that the minimum distance-$d$ dominating set in any graph is bounded by the maximum cardinality distance-$d$ independent set of the same graph. The same result holds for any sub-graphs also. 
\end{proof}
 \begin{lemma} \label{lem:2}
 The time complexity to compute the optimal D$d$DS in $\chi$ is $d^2n^{O(d)}$.
 \end{lemma}
 \begin{proof}
Consider a $d$-width region around a cell $\chi$ as $\chi'$ (see Figure \ref{fig:dom_approx}(c)), having point set $P_\chi'\subseteq P$. Let $G_{P_\chi'}$ be a weighted complete graph with the vertex set corresponding to point set $P_\chi'$ and edge costs are as defined in Section \ref{dd}. Apply all-pairs shortest path algorithm \cite{cormen} on graph $G_{P_\chi'}$ and store the result in a matrix $\cal M$.
 
 Observe that, for computing a MD$d$DS in graph $G_\chi$, we need to compute a MD$d$DS in each and every component $C_i\in G_\chi$. As per the definition of components, all the components are $d$-distance apart from each other. So, taking the union of the computed solutions of each component leads to a MD$d$DS for the graph $G_\chi$. For computing an optimum solution in a component $C_i$, we consider all possible tuples of size at most $O(d)$ (refer Lemma \ref{lem:apprx_cmpnt}) and check whether the selected tuple is a feasible solution or not with the help of matrix $\cal M$ in $O(d^2)$ time. So, MD$d$DS can be computed in a single component $C_i$ in $O(d^2|C_i|^{O(d)})$ time. Hence, computing an MD$d$DS in $G_\chi$ takes  $O(d^2\sum\limits_{C_i\in G_\chi}|C_i|^{O(d)})=d^2n^{O(d)}$ time, where $n$ is the number of vertices in $G_\chi$.
  \end{proof}
 \begin{theorem}
 Given a set $P$ of $n$ points in the plane $\cal R$, a distance-$d$ dominating set of size at most $4|OPT|$ can be computed in $d^2n^{O(d)}$ time, where $OPT$ is a minimum distance-$d$ dominating set.
 \end{theorem}
\begin{proof}
Follows from Lemma \ref{lem:1}, Lemma \ref{lem:apprx_cmpnt} and Lemma \ref{lem:2}.
\end{proof} 
 
 \section{Approximation scheme for GD$d$DS Problem} \label{sec:apprx_schm}
 
 In this section, using the technique of shifting strategy  \cite{hochbaum}, we propose a polynomial time approximation scheme (PTAS) for the D$d$DS problem, for a given constant $d$. Given a point set $P$ (centers of the UDG) in a rectangular region ${\cal R}$ and a fixed integer $k \gg d$.
 We use two-level nested shifting strategy as follows:
 \begin{itemize}
 \item first, we apply shifting strategy in the horizontal direction. The $i$-th iteration $(1\leq i\leq k)$ of the first level, partition ${\cal R}$ into horizontal strips such that (i) first strip is of width $i$, and (ii) remaining strips are of width $k$. Note that width of last strip may be less than $k$.
 \item without loss of generality, assume that each points lying on the left boundary of a strip belong to its left adjacent strip.
 \item consider each non-empty horizontal strip $H$, and apply second level of shifting strategy on the vertical direction.
 \item in the second level, the $j$-th iteration $(1\leq j\leq k)$ partition each non-empty horizontal strip $H$ into square/rectangular cells of size (i) $j\times \ell$ for the first cell, and (ii) $k\times \ell$ for all other cells, where $\ell$ defines the width of the strip $H$ ($\ell=i$ for the first strip and $\ell=k$ for all other strips).
  \end{itemize}
 
We solve each $k \times k$ square (conceptually extend the smaller cells to  $k \times k$ square) optimally. Take the union of each $k \times k$ square in a horizontal strip to get a feasible solution of each  strip. Finally, we take the union of solutions of  each non-empty horizontal strip to get a feasible solution of the problem in a single iteration. In the same process, we get the feasible solutions of all the iterations in the first level. We report the solution $D$, having minimum cardinality among all the solutions generated in each iterations as the solution of  the D$d$DS problem.

  Now, we discuss the process of getting solution from each $k \times k$ square optimally. Before discussing the process, we compute a matrix $\cal M$ containing the cost of all pair shortest paths in a complete graph defined with the points in $P$ where the edge costs are as defined in Section \ref{dd}.
 \subsection{Computing an optimum solution in a $k \times k$ square}
We apply the same strategy as described in Sub-section \ref{subsec:1}  on the point set $P_\chi\subseteq P$ inside a square $\chi$ of size $k\times k$. Let $P'_\chi\subseteq P_\chi$ denote the point set which are at most $d$ distance away from $\ell_h$ and $\ell_v$ (the horizontal and vertical lines which divides $\chi$ into $4$ squares). Let $P''_\chi$ be a minimum cardinality subset of $P'_\chi$ such that all the points in $P'_\chi$ are distance-$d$ dominated by the point set $P''_\chi$.

\begin{lemma}
$|P''_\chi|\leq O(k)$
\end{lemma}
 \begin{proof}
 Follows from Lemma \ref{lem:apprx_cmpnt} with some similar combinatorial argument.  
 \end{proof}
 
 We use similar divide and conquer strategy as discussed in the Sub-section \ref{subsec:1} for obtaining the optimum solution on $\chi$.

\begin{lemma} \label{lem_complx}
The optimum solution produced by our algorithm for each $k\times k$ square ($\chi$) takes $k^2n_{\chi}^{O(k)}$ time, where $n_\chi=|P_\chi|$ is the number of points inside $\chi$.
\end{lemma}
\begin{proof}
As our algorithm checks all combinations of points of size $|OPT_\chi|$, where $OPT_\chi$ is an optimal solution for  $\chi$, there must be a case that the combination of points in $OPT_\chi$ appear in the process.

The time complexity follows from Lemma \ref{lem:apprx_tc}.
\end{proof}

\begin{theorem}
Given a set $P$ of $n$ points (center of the unit disks) in ${\cal R}$ and an integer $k \gg d$, a distance-$d$ dominating set of size at most $(1+\frac{1}{k})^2 \times |OPT|$ can be computed in $k^2n^{O(k)}$ time, where $OPT$ is the optimum solution.
\end{theorem}
 \begin{proof}
  Let $OPT$ be a minimum D$d$DS for the point set $P$ in UDG $G$, and $OPT'\subseteq OPT$ be the points chosen in $OPT$, which $d$-distance dominates the points outside the boundary  of all the cells in an iteration (first level $i$-th iteration and second level $j$-th iteration). Let $D^*$ be a solution obtained by our algorithm in an iteration.

  Then, $|D^*|\leq |OPT|+|OPT'|$. For all the iterations of $(i,j)$ ($1\leq i,j \leq k$), we have \\ $\sum\limits_{i=1}^{k}\sum\limits_{j=1}^{k} |D^*|\leq k^2|OPT|+\sum\limits_{i=1}^{k}\sum\limits_{j=1}^{k}|OPT'|$.\\
 Since any point from a cell $\chi$ chosen in $OPT$ can $d$-distance dominate points from no more than one horizontal strip (or vertical strip), and at most $k$ times each horizontal (or vertical) boundary appears throughout the algorithm, we have\\
 $\sum\limits_{i=1}^{k}\sum\limits_{j=1}^{k}|OPT'|\leq k|OPT|+k|OPT|$.\\
 Thus,\\
 $\sum\limits_{i=1}^{k}\sum\limits_{j=1}^{k}|D^*|\leq k^2|OPT|+2k|OPT|=(k^2+2k)|OPT|$.\\ Thus, $min\sum\limits_{i=1}^{k}\sum\limits_{j=1}^{k}|D^*|\leq (1+\frac{1}{k})^2 \times |OPT|$.
\\ the time complexity result follows from Lemma \ref{lem_complx}. 
 \end{proof}

\section{Conclusion} \label{sec:conclusion}

In this article, we studied the D$d$IS problem, a variant and generalized version of the independent set problem and D$d$DS problem, a generalized version of dominating set problem, on unit disk graphs. We proved that both GD$d$IS and GD$d$DS problems are NP-complete on unit disk graphs. We proposed a simple 4-factor approximation algorithms for both the problems.  We also proposed polynomial time approximation schemes (PTAS) for each of the problems.
\bibliographystyle{spmpsci}
\bibliography{JoCO}
\end{document}